\documentclass[conference]{IEEEtran}
\IEEEoverridecommandlockouts

\usepackage{amsmath,amssymb,amsfonts,amsthm}
\numberwithin{equation}{section}

\theoremstyle{plain} 
\newtheorem{theorem}{Theorem}[section]
\newtheorem{lemma}[theorem]{Lemma}
\newtheorem{proposition}[theorem]{Proposition}

\theoremstyle{remark} 

\theoremstyle{definition} 
\newtheorem{definition}[theorem]{Definition}

\newtheorem*{theorem*}{Theorem}
\newtheorem*{lemma*}{Lemma}
\newtheorem*{proposition*}{Proposition}
\newtheorem*{corollary*}{Corollary}
\newtheorem*{remark*}{Remark}
\newtheorem*{definition*}{Definition}

\usepackage{algorithm}
\usepackage{algpseudocode}

\usepackage{graphicx}
\usepackage{subcaption}
\usepackage{booktabs}
\usepackage{multirow}
\usepackage[table,xcdraw]{xcolor}
\graphicspath{{./images/}}

\usepackage[colorlinks=true, linkcolor=blue, citecolor=red, urlcolor=purple]{hyperref}
\usepackage[capitalise]{cleveref}

\usepackage{times}
\usepackage{helvet}
\usepackage{courier}
\usepackage{tcolorbox}
\usepackage{soul}
\usepackage{dsfont}
\usepackage{comment}
\usepackage{textcomp}

\newcommand{\DP}{differential privacy }

\newcommand{\Prob}{\mathbb P}

\newcommand{\subhead}[1]{\vspace{0.04in}\noindent\textbf{#1.}}

\begin{document}

\title{\textsc{Lap}\textsubscript{2}: Revisiting Laplace DP-SGD for High Dimensions via Majorization Theory}


\author{
\IEEEauthorblockN{Meisam Mohammady\textsuperscript{$\dagger$}, Qin Yang\textsuperscript{$\P$}, Nicholas Stout\textsuperscript{$\dagger$}, Ayesha Samreen\textsuperscript{$\dagger$}, Han Wang\textsuperscript{$\S$}, \\Christopher J Quinn\textsuperscript{$\dagger$}, Yuan Hong\textsuperscript{$\P$}}
\IEEEauthorblockA{
\textsuperscript{$\dagger$}Iowa State University, \textsuperscript{$\P$}University of Connecticut, \textsuperscript{$\S$}University of Kansas 
}
}

\maketitle

\begin{abstract}
Differentially Private Stochastic Gradient Descent (DP-SGD) is a cornerstone technique for ensuring privacy in deep learning, widely used in both training from scratch and fine-tuning large-scale language models. While DP-SGD predominantly relies on the Gaussian mechanism, the Laplace mechanism remains underutilized due to its reliance on $\ell_1$ norm clipping. This constraint severely limits its practicality in high-dimensional models because the $\ell_1$ norm of an $n$-dimensional gradient can be up to $\sqrt{n}$ times larger than its $\ell_2$ norm. As a result, the required noise scale, and thus the privacy loss, grows significantly with model size, leading to poor utility or untrainable models.

In this work, we introduce \textsc{Lap}\textsubscript{2}, a new solution that enables $\ell_2$ clipping for Laplace DP-SGD while preserving strong privacy guarantees. We overcome the dimensionality-driven clipping barrier, by computing coordinate-wise moment bounds and applying \textit{majorization theory} to construct a tight, data-independent upper bound over the full model. By exploiting the \textit{Schur-convexity} of the moment accountant function, we aggregate these bounds using a carefully designed majorization set that respects the $\ell_2$ clipping constraint. This yields a multivariate privacy accountant that scales gracefully with model dimension and enables the use of thousands of moments. Empirical evaluations demonstrate that our approach significantly improves the performance of Laplace DP-SGD, achieving results comparable to or better than Gaussian DP-SGD under strong privacy constraints. For instance, fine-tuning RoBERTa-base (125M parameters) on SST-2 achieves 87.88\% accuracy at $\mathbf{\epsilon = 0.54}$, outperforming Gaussian (87.16\%) and standard Laplace (48.97\%) under the same budget. 
\end{abstract}


\section{Introduction}



Training and fine-tuning deep learning models pose significant privacy risks. During training, adversaries can exploit gradient updates, model outputs, and pre-trained parameters to reconstruct sensitive data, making privacy a critical concern. Moreover, fine-tuning large pre-trained language models, such as BERT \cite{liu2019roberta} and GPT families \cite{yu2021differentially}, is essential for achieving state-of-the-art performance in various tasks, including sentence classification \cite{liu2019roberta}, text generation \cite{novikova2017e2e}, and code generation \cite{wang2018glue}. Data reconstruction attacks, such as Updates-Leak \cite{247690} and Inverting Gradients \cite{10.5555/3495724.3497145}, achieve success rates of up to 80\%, while the Dynamic Memory Model Inversion Attack (DMMIA) \cite{10.1145/3581783.3612072} enhances realism, reaching 93.54\% on FaceScrub. Additionally, prompt-based techniques further expose LLMs to training data extraction attacks, increasing the risk of recovering individual samples \cite{carlini2019}.

To mitigate privacy risks during model training and fine-tuning, differential privacy (DP) \cite{dwork2006differential} has become the de facto privacy model. A widely adopted method, Differentially Private Stochastic Gradient Descent (DP-SGD) \cite{abadi2016deep}, provides DP guarantees by ensuring that the inclusion or exclusion of any single data sample does not significantly impact the model's output. DP-SGD achieves this by clipping gradients to limit the influence of individual samples and adding Gaussian noise to the gradients within each batch. While DP-SGD effectively controls the privacy budget consumed over the thousands of iterations typically required for model training or fine-tuning, it often results in noticeable utility loss \cite{abadi2016deep,balle2018improving,gopi2021numerical}.


However, Gaussian noise exhibits the \textbf{\textit{privacy wall}} phenomenon~\cite{mironov2019r, TAN}, where DP-SGD with Gaussian perturbations fails to attain the optimal privacy–utility scaling $\sigma = \Theta(1/\varepsilon)$ in the medium privacy regime, resulting in suboptimal trade-offs. This occurs because the moments accounting function (MAF) for Gaussian noise follows a quadratic exponential form, causing accounted privacy loss to escalate rapidly as \(\lambda\) increases. Then, Gaussian-based DP-SGD may degrade model accuracy and slow convergence to some extent, particularly in \textbf{high-privacy regimes} (small \(\epsilon\)).

Alternatively, Laplace mechanism \cite{holohan2018bounded} has been shown to preserve accuracy better than the Gaussian mechanism in stricter privacy regimes, particularly in low-dimensional settings \cite{geng2015optimal, muthukrishnan2025differential}. This advantage could be highly beneficial for deep learning, by ensuring strong privacy guarantees across many training iterations. However, this improvement has not been fully realized, likely due to the destructive impact of $\ell_1$-norm clipping in high-dimensional spaces (DP-SGD with Laplace mechanism requires it by default). This limitation arises because, for an $\mathbf{n}$-dimensional gradient vector, the $\ell_1$ norm can be up to $\mathbf{\sqrt{n}}$ times larger than its $\ell_2$ norm. 

\vspace{-0.15in}

\begin{figure}[!h]
    \centering
    \includegraphics[width=0.8\linewidth, trim=20 210 20 180, clip]{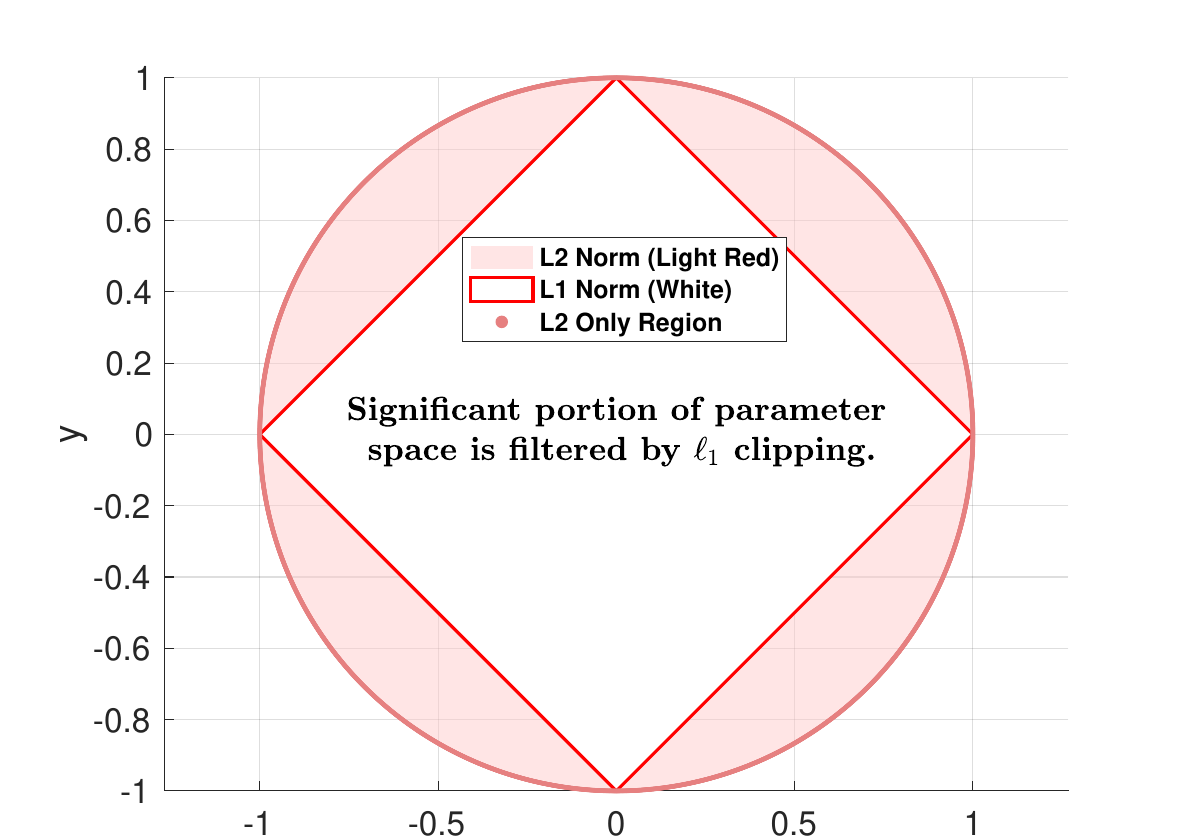}
    \caption{Comparison of $\ell_1$ and $\ell_2$ norm clipped spaces. 
    }
     \label{fig:clip}
\end{figure}

Both the Laplace and Gaussian mechanisms can be applied to vector-valued functions, but they differ fundamentally in how they interact with gradient clipping. The Gaussian mechanism is based on $\ell_2$ sensitivity and is naturally compatible with $\ell_2$ norm clipping, preserving a large feasible region for optimization—especially in high-dimensional settings. In contrast, the Laplace mechanism requires $\ell_1$ norm clipping due to its dependence on $\ell_1$ sensitivity. As illustrated in Figure~\ref{fig:clip}, this severely restricts the optimization region and leads to overly conservative updates, often degrading training performance under strong privacy constraints.

To resolve this mismatch, we propose a new privacy accounting framework for Laplace DP-SGD that supports $\ell_2$ norm clipping. A natural workaround (summing coordinate-wise privacy losses) can be overly pessimistic in high-dimensional settings, as it ignores the global $\ell_2$ constraint and accumulates leakage too conservatively. 

Our key insight is to instead apply \textit{majorization theory}~\cite{marshall1979inequalities}, a classical tool from inequality analysis, to derive a tighter, dimension-aware privacy bound. When the moments accountant is \textit{Schur-convex}, as in the case of the Laplace mechanism, we can replace the actual vector of per-parameter gradient magnitudes with a carefully constructed \textit{majorization set}, a worst-case configuration that dominates all valid $\ell_2$-clipped gradients. This allows us to compute moment bounds for each coordinate and sum them over the majorization set instead of the data-dependent gradients.

This yields a tight, dimension-aware upper bound on privacy loss that remains data-independent and scales gracefully with model size. As a result, our method, \textsc{Lap}\textsubscript{2}, enables high-utility training of large models such as RoBERTa-Large and ViT under strong privacy budgets (e.g., $\epsilon \leq 1$), making the Laplace mechanism a practical choice once again for modern DP-SGD workloads.

Beyond the theoretical contributions, we introduce a framework that enables practitioners to integrate \textsc{Lap}\textsubscript{2} into various AI applications by adapting noise parameters to task-specific needs. Our framework allows users to systematically compute the optimal Laplace noise scale ($b$) and clipping parameter ($C$) based on: (1) AI task specifications (epochs, batch size, model size), and (2) DP constraints ($\epsilon, \delta$). 

Therefore, the main contributions are summarized as below: 
\begin{enumerate}
    \item Taking the first step to advance the \textit{Laplace mechanism in DP-SGD} by mitigating its reliance on $\ell_1$ clipping through \textit{majorization theory}, enabling tighter privacy accounting under fixed distortion. 

    
    \item Introducing \textsc{Lap}\textsubscript{2} that offers \textit{plug-and-play DP-SGD} where users can seamlessly compute optimal clipping and noise parameters for their setups.

    \item Theoretically analyzing \textit{\textsc{Lap}\textsubscript{2}}'s privacy and performance, and comprehensively comparing the \textit{\textsc{Lap}\textsubscript{2}} and \textit{Gaussian} mechanisms in DP-SGD. 


    \item Conducting comprehensive empirical evaluations on both vision and language models, demonstrating that \textit{\textsc{Lap}\textsubscript{2}} consistently performs comparably to the Gaussian mechanism, with higher accuracy in certain strong-privacy regimes and larger model fine-tuning.
\end{enumerate}
Our evaluation code and additional experimental results are publicly available at \href{https://github.com/datasec-lab/lap2}{the LAP$_2$ evaluation repository}.
\section{Preliminaries}
\label{sec:prem}
We review some background on \DP and DP-SGD for the theoretical foundations of the \textsc{Lap}\textsubscript{2} framework.

\subsection{Differential Privacy}\label{section: DP def}
\label{def: differential privacy original}

Differential privacy ensures individual input privacy by introducing randomness into the output, 
either by injecting explicit noise or by leveraging inherent stochasticity in the mechanism or inputs%
\footnote{Developing a unified framework to capture all randomness sources and translate them into differential privacy guarantees remains an open challenge \cite{mironov2019r}. Here, we assume a fixed probability space \((\Omega, \mathcal{F}, P)\), where \(\Omega\) is the sample space, \(\mathcal{F}\) is a \(\sigma\)-algebra of events, and \(P\) is the probability measure.}
. 
Let $\mathcal{D}$ denote the space of datasets of interest.
A mechanism can be viewed as a probabilistic 
algorithm designed to answer a query $q$, which is a map  $q: \mathcal{D} \rightarrow \mathbb{R}^n$. We sometimes index the mechanism by the query $q$ of interest, denoting it as $M_q$. 
In particular, we denote mechanisms answering queries $q: \mathcal{D} \rightarrow \mathbb{R}$ 
by $M_q(\mathcal{D})$.
Additionally, we define a symmetric binary relation, $\mathrm{Adj}$, on $\mathcal{D}$, called \emph{adjacency}~\cite{dwork2006calibrating}. 
Two datasets $d \in \mathcal{D}$ and $d' \in \mathcal{D}$ are adjacent, denoted $\mathrm{Adj}(d, d')$, if and only if they differ by the data of a single participant.
The standard definitions of differential privacy introduced in~\cite{DKM+06,DMNS06} will follow.
\begin{definition} [Differential Privacy \cite{Dwork06}]
A randomization mechanism $M_q: \mathcal{D} \times \Omega 
\to \mathbb{R}^n$ which is $\epsilon$-differentially private, necessarily
randomizes its output in such a way that for all adjacent datasets $d,\ d' \in \mathcal{D}$ and $S \in \mathbb{R}^n$, 
\begin{align}\label{eq: standard def approximate DP original}
\Prob(M_q(d) \in S) \leq e^{\epsilon} \Prob(M_q(d') \in S). \;\; 
\end{align}
\end{definition}

 If the inequality fails, an \(\epsilon\)-breach occurs, indicating 
 a non-negligible probability to distinguish the presence or absence
 of any single $d_n=d-d'$.
 The following norms will be frequently used throughout the remainder of the paper.

 \begin{definition}
The \(\ell_1\) and \(\ell_2\) norms of a vector \( x \in \mathbb{R}^n \) are defined as:
\(
\|x\|_1 = \sum_{i=1}^{n} |x_i|, \quad
\|x\|_2 = \left( \sum_{i=1}^{n} x_i^2 \right)^{\frac{1}{2}}
\)
where \( x = (x_1, x_2, \dots, x_n) \) is a vector in \( \mathbb{R}^n \).
\end{definition}

 We now recall a fundamental mechanism that provides \(\epsilon\)-differential privacy when answering queries.

\begin{definition}[Laplace Mechanism \cite{dwork2006our}]
\label{def:LapMech}
Given a numerical query \(q(d)\) (whose output lies in \(\mathbb{R}^n\)) and a scale \(b\), 
the Laplace mechanism \(M_q(d,b)\) modifies \(q(d)\) by adding noise 
\(w \sim \mathrm{Lap}(b)\), i.e., \(M_q(d,b) = q(d) + w\).
\end{definition}

Recall that the (univariate) Laplace distribution with mean zero and scale \(b\), denoted \(\mathrm{Lap}(b)\),
has density \(p(x; b) = \tfrac{1}{2b}\,\exp\bigl(-|x|/b\bigr)\) and variance \(2b^2\).
For \(w \in \mathbb{R}^n\) with i.i.d.\ components \(w_i \sim \mathrm{Lap}(b)\), 
the joint density is \(\bigl(\tfrac{1}{2b}\bigr)^n \exp(-\|w\|_1/b)\) and $\|\cdot\|_1$ denotes the  $\ell_1$ norm \cite{Walker1965ProbabilityTA}.


The scale parameter in the Laplace mechanism determines the degree of privacy. Specifically:

\begin{theorem}[Laplace Mechanism \cite{Dwork10}]
\label{thm:LapMech}
Let \(q : \mathcal{D} \to \mathbb{R}\) be a query with (global) $\ell_1$-sensitivity 
$ \Delta_1 q
  \;=\;
  \max_{d,d' : \mathrm{Adj}(d,d')} 
  \,\bigl\|\,q(d) - q(d')\bigr\|_{1}$.
If the Laplace mechanism \(M_q(d, b)\) uses a scale parameter 
\(b \ge \tfrac{\Delta q}{\epsilon}\), 
then \(M_q(d, b)\) is \(\epsilon\)-differentially private.
\end{theorem}

Following the introduction of differential privacy~\cite{DMNS06, Dwork06}, various relaxations were proposed to capture the properties of Gaussian additive noise. Among these, approximate DP introduced an additive \(\delta\) term for Gaussian noise analyses which is typically of the order of $\frac{1}{|\mathcal{D}|}$.


\begin{definition}[Approximate Differential Privacy~\cite{dwork2006calibrating}]
\label{def:epsDP}
 A randomization mechanism \(M_q: \mathcal{D} \to \mathcal{R}\) is said to provide \((\epsilon, \delta)\)-DP for releasing the query results \(q(\mathcal{D})\) if it randomizes its output such that, for any two adjacent datasets $d,\ d' \in \mathcal{D}$  and all subsets \(S \subseteq \mathbb{R}^n\), the following holds:
\begin{align}
\label{def:DP}
\Prob(M_q(d) \in S) \leq e^{\epsilon} \Prob(M_q( d') \in S) + \delta,
\end{align}
where $\delta$ represents the failure probability of the $\epsilon$-DP guarantee provided by the randomization mechanism.
\end{definition}

\begin{definition}[Gaussian Mechanism \cite{dwork2006calibrating}]
    The Gaussian mechanism $M_q(d, \sigma)$ modifies the answers to query $q$ in the dataset $d$ by adding noise $w \sim \mathcal{N}(0, \sigma^2 I_n)$, that is, $M_q(d, \sigma) = q(d) + w$. 
\end{definition}
The Gaussian mechanism ensures approximate differential privacy. In particular, for 
\(\epsilon < 1\), its standard deviation is given by $\sigma \;=\; \frac{\Delta_2 q}{\epsilon} \,\sqrt{2\,\ln\bigl(\tfrac{1.25}{\delta}\bigr)}$, where \(\Delta_2 q\) is the \(\ell_2\)-sensitivity of the query \(q\) across all adjacent datasets.

Sequentially applying DP mechanisms increases the overall privacy cost, as shown by
various composition theorems~\cite{dwork2010boosting,7883827}. Naive composition states that
chaining \(k\) mechanisms, each \((\epsilon,\delta)\)-DP, results in an overall
\((k\epsilon, k\delta)\)-DP guarantee, which can be overly conservative.
Stronger composition theorems, such as the \emph{advanced composition} theorem~\cite{kairouz2015composition},
provide tighter bounds. In particular, the overall DP guarantee can become $\Bigl(\epsilon \sqrt{2k \,\ln\! \ \bigl(\tfrac{1}{\delta}\bigr)} 
  \;+\; k\,\epsilon\,\bigl(e^\epsilon - 1\bigr)\Bigr)$.



\subsection{Deep Learning with Differential Privacy}
\label{sec:DP-DPSGD}
DP-SGD~\cite{abadi2016deep} is the first method that integrates the guarantees of differential privacy (DP) into deep neural network (DNN) training. At each iteration \(t\), the gradient \(\mathbf{g}_t(\mathbf{x})\) for a data point \(\mathbf{x}\) is clipped using a threshold \(C\). Formally, the clipped gradient is defined as:
\[
\mathbf{g}_C(\mathbf{x})= \frac{\mathbf{g}_t(\mathbf{x})}{\max\Bigl(1, \tfrac{\|\mathbf{g}_t(\mathbf{x})\|_2}{C}\Bigr)}.
\]
This clipping step ensures that the sensitivity of each gradient is bounded with respect to the inclusion of individual samples in the training set, thereby preparing it for the perturbation with a Gaussian mechanism. 
\begin{equation}
\label{eqn:gradientperturb}
    \tilde{\mathbf{g}}(\mathbf{x})= \mathbf{g}_C(\mathbf{x}) + \mathcal{N}(0, C^2 \sigma^2 \mathbf{I}_n).
\end{equation}
The final update \(\tilde{\mathbf{g}}\) is computed by averaging \( \tilde{\mathbf{g}}(\mathbf{x}) \) over the batch of size \(L\). Even with tighter composition bounds, training over many iterations (e.g., thousands of rounds) can lead to high cumulative privacy loss (\(\epsilon\)). 
DP-SGD mitigates this by formulating an accounting function over the privacy loss terms across rounds, namely the \emph{Moments Accountant Function}. Consequently, DP-SGD with moments accounting achieves a significantly improved bound of $\mathcal{O}(\epsilon \sqrt{T} L/|\mathcal{D}| , \delta)$-DP. 

\subhead{Moments Accounting Function}
\label{sec:MomentAccounting}
Consider two neighboring datasets \(d,d'\in \mathcal{D}\) and an outcome \(o\in\mathbb{R}^n\) 
from the sanitized gradients \(\tilde{\mathbf{g}}(\cdot)\). 
The \emph{privacy loss} w.r.t. \(o\) is defined by
\begin{align}
\label{eq:PrivLoss:c}
  c(o; \text{aux}, d, d') 
  \;=\; 
  \log\! \ \Bigl[
    \frac{\Prob(\tilde{\mathbf{g}}(\text{aux}, d)=o)}{\Prob(\tilde{\mathbf{g}}(\text{aux}, d')=o)}
  \Bigr].
\end{align}

DP-SGD ensures privacy by focusing on orthogonal updates relative to the ``aux'' state, which represents the model's state \textbf{prior to} training on a specific data instance 
$x_{i<N}$. This instance has been sequentially updated through the iterative application of differentially private mechanisms \cite{abadi2016deep}.

%
%
\color{black}
Let $\mu_0$ denote the probability density function
(pdf) of $\mathcal{N}(0, \sigma^2)$, and $\mu_1$ denote the pdf of $\mathcal{N}(1, \sigma^2)$. 
Let $\mu$ denote the mixture of two Gaussians, $\mu = (1-\zeta)\mu_0 + \zeta \mu_1$, where $\zeta=L/\max\{d,d'\}$ 
denotes the batch sampling rate.
Clipping with threshold $C=1$, the privacy loss of DP-SGD with Gaussian noise is bounded \cite{abadi2016deep} by:
%
\begin{align}
\label{eq:PrivLoss:c}
  c(o; \text{aux}, d, d') 
  \;\leq\; 
  \log\! \ \Bigl[ \max{ \left\{
    \frac{\mu_0(o)}{\mu(o)} , \frac{\mu(o)}{\mu_0(o)} \right\}}
  \Bigr].
\end{align}

DP-SGD separates clipping from privacy accounting. 
To incorporate the impact of the clipping parameter $C$, DP-SGD scales the
Gaussian noise's variance by $C^2$ per Equation \ref{eqn:gradientperturb} after accounting, without directly influencing privacy accounting itself. To bound the privacy loss more tightly, the moments of the random variable associated with \( c(o; \text{aux}, d, d') \) are tightly calculated and bounded. Precisely, for \(\lambda>0\), define
%
%

\begin{align}
\label{eq:lambda-moment-gen}
  \alpha_{M}(\lambda;\mathrm{aux},d,d') 
  = \log\,\mathbb{E}_{o \sim M(\mathrm{aux},d)}
    \!\bigl[\exp\!\bigl(\lambda\,c(o;\mathrm{aux},d,d')\bigr)\bigr].
\end{align}

We then take the worst-case over \(\mathrm{aux}, d, d'\):
\begin{align}
\label{alpha func}
  \alpha_{M}(\lambda)
  \;=\;
  \max_{\mathrm{aux},\,d,\,d'} 
 \Bigl( \alpha_{M}(\lambda;\mathrm{aux},d,d') \Bigr).
\end{align}

Key properties of this \textit{Moments Accounting Function (MAF)} include \cite{abadi2016deep}:

\begin{enumerate}
    \item \textbf{Composability.} 
    For a sequence of adaptive mechanisms \(M_1,\dots,M_t\),
     \begin{align}
    \label{composmaf}
      \alpha_{M_1,\dots,M_t}(\lambda) 
      \;\le\; 
      \sum_{i=1}^t 
      \alpha_{M_i}(\lambda).
    \end{align}
    \item \textbf{Tail Bound.}
    For any \(\epsilon>0\), a sequence of adaptive mechanisms \(M_1,\dots,M_t\)
    is jointly \((\epsilon,\delta)\)-DP for  
    \begin{align}
    \label{tail}
      \delta 
      \;=\;
      \min_{\lambda>0} 
     \Bigl( \exp\! \ \bigl( 
     \alpha_{M_1,\dots,
     M_t}
     (\lambda) \;-\;\lambda\,\epsilon\bigr)\Bigr).
    \end{align}
\end{enumerate}

\subhead{Privacy Amplification with Subsampling} While the MAF significantly tightens the privacy budget, 
DP-SGD's main strength arises from \emph{privacy amplification via random sampling}, 
where the privacy cost per evaluation decreases \emph{quadratically}—rather than linearly—with 
the sampling rate
\(\zeta = \tfrac{L}{|d|}\).
%
For instance,  Bun et al.~
\cite{BDRS18} showed that 
\(\alpha(\lambda) \propto \zeta^2 \,\tfrac{6\,\lambda}{\sigma^2}\). Mironov et al.~
\cite{mironov2019r} (Section 3.3) presented a tighter privacy amplification bound, which remains the state-of-the-art for the subsampled Gaussian mechanism. The bound is expressed as:
\begin{equation}
\label{A_alpha}
\alpha_{M_1,\dots,M_t}(\lambda) 
\leq t \cdot \log  \sum_{\eta=0}^{\lambda+1} \binom{\lambda}{\eta} (1 - \zeta)^{\lambda+1 - \eta} \zeta^{\eta} \cdot \exp\left(\frac{\eta^2 - \eta}{2\sigma^2}\right).
\end{equation}

 The final \((\epsilon, \delta)\)-DP guarantee is derived using the tight conversion formula provided by Balle et al.~\cite{pmlr-v108-balle20a}:

\small
\begin{equation}
    \epsilon(\delta) = \min_{\lambda > 0} \Bigl(\frac{\alpha_{M_1,\dots,M_t}(\lambda)}{\lambda} + \log\left(\frac{\lambda}{\lambda + 1}\right) - \frac{\log(\delta) + \log(\lambda + 1)}{\lambda}\Bigr). \nonumber
\end{equation}
\normalsize

\noindent
This significantly enhances the analysis of privacy loss in scenarios involving repeated subsampling and Gaussian noise. 

\section{Problem Statement}
\label{sec:methodology}

DP-SGD predominantly relies on the Gaussian mechanism, with Laplace noise only applied in limited cases \cite{Sommer2019PrivacyLC, 10.1145/3447548.3467268}. The fundamental limitation of the Laplace mechanism in these settings stems from its privacy loss, which requires $\ell_1$-clipping. This clips gradients more aggressively than $\ell_2$-clipping, especially in higher dimensions, resulting in less informative updates and making it less suitable for DP-SGD. 


To formalize this difference, consider the well-known relationship between $\ell_1$ and $\ell_2$ norms: $\|x\|_1 \leq \sqrt{n} \|x\|_2$, which follows from the Cauchy–Schwartz inequality~\cite{steele2004cauchy}. The difference in retained volume after $C$-$\ell_1$ vs. $C$-$\ell_2$ clipping corresponds to the ratio of the respective clipped spaces in $\mathbb{R}^n$: a cross-polytope with volume $V_{\ell_1}(n, C) = \frac{(2C)^n}{n!}$ and a Euclidean ball with volume $V_{\ell_2}(n, C) = \frac{\pi^{n/2} C^n}{\Gamma(n/2 + 1)}$. Their ratio simplifies to  

\[
\frac{V_{\ell_1}(n, C)}{V_{\ell_2}(n, C)} = \left(\frac{2}{\sqrt{\pi}}\right)^n \frac{(\frac{n}{2})!}{n!},
\]

which shrinks exponentially as $n$ increases, highlighting the severe loss of preserved vectors under $\ell_1$ clipping. 

Figure~\ref{fig:l1vsl2vol} depicts the volume differences between $C$-$\ell_1$ and $C$-$\ell_2$ norm balls across increasing dimensionality. As the number of model parameters $n$ increases, these volumes diverge at an astronomical rate, making $\ell_1$ clipping significantly more restrictive than $\ell_2$ clipping in high dimensions by discarding a much larger proportion of gradients. Even with larger clipping thresholds, this geometric gap remains substantial for $n > 20$.

\begin{figure}[ht] 
    \centering
    \includegraphics[width=0.7\linewidth, trim=40 180 70 212, clip]{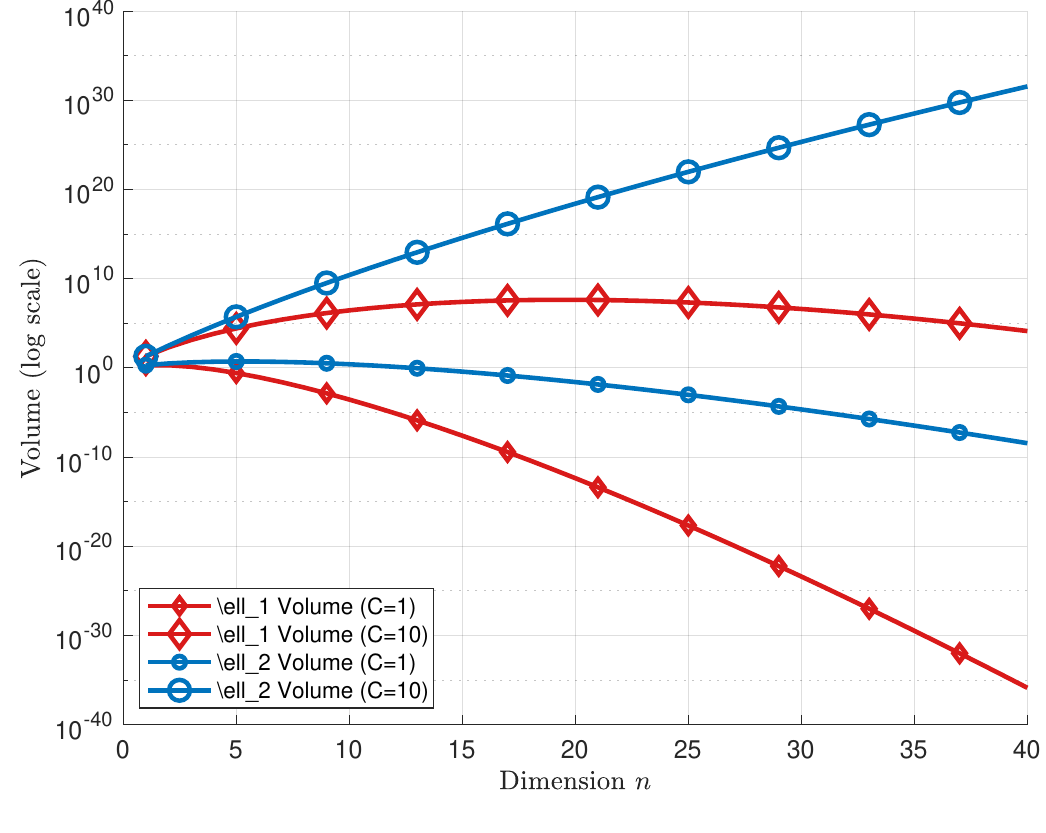}
    \caption{Volume of $n$-dimensional vector space clipped by $C$-$\ell_1$ and $C$-$\ell_2$ norms.}
    \label{fig:l1vsl2vol}
    \vspace{-10pt}
\end{figure}

This geometric bottleneck directly impacts the applicability of the Laplace mechanism for deep learning with DP. When coupled with $\ell_1$-norm clipping, the mechanism confines updates to a cross-polytope whose volume decays exponentially with $n$, making it unsuitable for high-dimensional gradients.

We formalize the privacy loss induced under this setup.

\begin{theorem}[Privacy Loss of a Laplace Mechanism]
\label{thm:lap2privacyloss}
The privacy loss of a Laplace mechanism with scale parameter $b$, applied to gradients clipped under the $\ell_1$ norm with threshold $C$, is bounded by $c(o; \text{aux}, d, d', b) \leq \frac{C}{b}$.
\end{theorem}

    
\begin{proof}
Let \( \tilde{\mathbf{g}}(\text{aux},\mathbf{x},b) \) denote the output of a Laplace mechanism with scale \( b \). The privacy loss at an outcome \( o \) is defined as:
\begin{align}
\label{eq:marginal}
    c(o; \text{aux}, d, d', b) 
    &= \log \frac{\Pr[\tilde{\mathbf{g}}(\text{aux},\mathbf{x},b) = o]}{\Pr[\tilde{\mathbf{g}}(\text{aux},-,b) = o]},
\end{align}
where \( \tilde{\mathbf{g}}(\text{aux}, -, b) \) denotes the outcome of the mechanism without access to the input \( \mathbf{x} \) (e.g. $\mathbf{x}$ is only in one data set). 
Thus, without loss of generality, the denominator follows a zero-mean Laplace PDF, as in the worst-case setting it lacks any gradient component in direction of $\mathbf{g}_C$, while the numerator is centered around \( \mathbf{g}_C\).
\begin{align}
& c(o; \text{aux}, d, d',b)= \log \frac{\text{Lap}(\mathbf{g}_C(\mathbf{x}), b I_n )}{\text{Lap}(0,b I_n)}\nonumber 
  \\
    %
    %
    &= \log \frac{ \left(\frac{1}{2b}\right)^n \exp\left(-\frac{\|o - \mathbf{g}_C(\mathbf{x}) \|_1}{b}\right)  }{ \left(\frac{1}{2b}\right)^n \exp\left(-\frac{\|o\|_1}{b}\right)   }= \frac{\|o \|_1 - \|o - \mathbf{g}_C(\mathbf{x})\|_1}{b} \nonumber.
\end{align}
For all real valued vectors $A$ and $B$, using $\|A\|_1 = \|(A-B)+B\|_1$, by the triangle inequality  we have:  $\|A\|_1 - \|B\|_1 \leq \|A-B\|_1$. Thus, \begin{equation}
\label{Laplacebound}
   \frac{\|o \|_1 - \|o - \mathbf{g}_C(\mathbf{x}) \|_1}{b} \leq \frac{\|\mathbf{g}_C(\mathbf{x})\|_1}{b} 
\end{equation}
Evaluating the  $\ell_1$ norm in terms of the elements of  
$\mathbf{g}_C(\mathbf{x})=[\mathbf{g}_1, \mathbf{g}_2,\cdots, \mathbf{g}_n]$, we have: 
\begin{equation}
\label{Laplacebound}
   c(o; \text{aux}, d, d', b) \leq \frac{\sum_{j=0}^n |\mathbf{g}_j|}{b}. 
\end{equation}
Thus, this completes the proof.
\end{proof}
While $c(o; \text{aux}, d, d', b) \leq \frac{C}{b}$ ensures bounded loss under $\ell_1$-clipping, applying it to gradients originally bounded in $\ell_2$ (e.g., $\|x\|_2 \leq C$) requires using the inequality $\|x\|_1 \leq \sqrt{n} \|x\|_2$. This inflates the effective $\ell_1$ clipping threshold to $\sqrt{n}C$, thereby scaling the privacy loss as $\frac{\sqrt{n}C}{b}$. As a result, naïvely using Laplace noise on $\ell_2$-clipped gradients leads to privacy degradation proportional to $\sqrt{n}$.

Thus, we pose the following problem.

\begin{center}
\begin{tcolorbox}[
    colback=white!97!gray,
    colframe=black!70,
    coltitle=white,
    colbacktitle=black,
    fonttitle=\bfseries\sffamily,
    title=Problem Statement,
    boxrule=0.8pt,
    width=0.46\textwidth,
    sharp corners=south
]
\rmfamily
Can we design a Laplace-based mechanism $\mathcal{M}$ that operates over $\ell_2$-clipped vectors, i.e., $x \in \mathbb{R}^n$ with $\|x\|_2 \le C$, while avoiding the $\sqrt{n}$ privacy cost overhead induced by $\|x\|_1 \le \sqrt{n}\|x\|_2$?  

That is, can
\[
\mathcal{M}(x) \sim \mathrm{Lap}_{B_2(C)}(x)
\]
achieve strong $(\epsilon,\delta)$-privacy without worst-case $\sqrt{n}$ degradation?
\end{tcolorbox}
\end{center}
\subsection{DP-SGD with a \textsc{Lap}\textsubscript{2} Mechanism}
To address the dimensional sensitivity problem stated earlier, we first present the moments accountant function for a subsampled uni-variate Laplace mechanism—where $\ell_1$ and $\ell_2$ clipping behave similarly at the coordinate level.

See Appendix~\ref{subsampledLap2} for the proof.

\begin{theorem}[Subsampled Uni-variate Laplace Mechanisms]
\label{MAflaplace}
Let $M$ be a uni-variate Laplace mechanism with scale parameter $b$ and sampling probability 
$\zeta = \frac{L}{N}$,
where $L$ is the mini-batch size and 
$N$ 
is the dataset size. Suppose $M$ is applied to a single 
partial derivative query (i.e. with respect to a single model parameter)
$q(d)=\mathbf{g}(d)$ clipped by threshold $C$, i.e., $|\mathbf{g}| \leq C$.

Then, the moments accountant function of $M_{q}$ satisfies
\begin{equation}
\label{lap_2account1}
\alpha_{M_{q}}(\lambda) = \log \left[
\sum_{\eta = 0}^{\lambda + 1} \binom{\lambda + 1}{\eta} (1-\zeta)^{\lambda + 1 - \eta} \zeta^\eta F(C, \eta)
\right],
\end{equation}

where
\begin{equation}
\label{eqn:G}
F(C, \eta) = \frac{\eta e^{\frac{(\eta-1)C}{b}} + (\eta-1) e^{-\frac{\eta C}{b}}}{2\eta-1}.
\end{equation}
\end{theorem}



While the uni-variate analysis in Equation~\ref{lap_2account1} provides an exact expression for each parameter, directly summing over millions of parameters can still lead to significant overestimation of the total privacy loss unless a tight holistic (multi-variate) bound is applied.%
\footnote{Here, the summation is over model parameters rather than training/fine-tuning iterations, similar to how composability traditionally applies over iterations.  
For the overall training/fine-tuning process, we would also need to compose over iterations.}

To address this overestimation, we apply \emph{Majorization Theory}~\cite{marshall1979inequalities,b81119c4-48ac-33d6-b3b8-e1de6925b554}, a powerful tool for comparing vectors based on their spread or concentration. After $\ell_2$ clipping with threshold $C$, the marginal clipped gradients $|\mathbf{g}_i|$ are mostly small but unevenly distributed. To capture this structure, we construct a \textit{majorization set}—a specially ordered vector that dominates the original gradient vector in a formal sense.

This majorization set enables us to bound any Schur-convex function of the original gradients by its value on the majorization set. Since the moments accountant function depends on the gradients through a Schur-convex structure, applying majorization allows us to obtain a tight holistic (multi-variate) privacy bound for $\ell_2$-clipped gradients.

In what follows, we first introduce background on majorization theory, then show that the moments accountant function given in Equation~\ref{lap_2account1} is Schur-convex. Finally, we derive a tight majorization set and use it to bound the total moment accountant function for Laplace mechanisms.

\begin{definition}[Weak Majorization]
Let \( \mathbf{x}, \mathbf{y} \in \mathbb{R}^n \). We say that \( \mathbf{x} \) \emph{weakly majorizes} \( \mathbf{y} \) from below, denoted \( \mathbf{x} \succ_w \mathbf{y} \), if
\[
\sum_{i=1}^{k} x_i^{\downarrow} \geq \sum_{i=1}^{k} y_i^{\downarrow} \quad \text{for all } k = 1, \dots, n,
\]
where \( x_i^{\downarrow} \) and \( y_i^{\downarrow} \) denote the components of \( \mathbf{x} \) and \( \mathbf{y} \), sorted in non-increasing order~\cite{marshall1979inequalities}.
\end{definition}

\begin{definition}[Strong Majorization and Schur-Convexity]
Let \( \mathbf{x}, \mathbf{y} \in \mathbb{R}^n \). We say that \( \mathbf{x} \) is \emph{majorized} by \( \mathbf{y} \), denoted \( \mathbf{x} \prec \mathbf{y} \), if
\[
\sum_{i=1}^{k} x_i^{\downarrow} \leq \sum_{i=1}^{k} y_i^{\downarrow} \quad \text{for all } k = 1, \dots, n, \quad \text{and} \quad \sum_{i=1}^{n} x_i = \sum_{i=1}^{n} y_i.
\]
A function \( F: \mathbb{R}^n \to \mathbb{R} \) is called \emph{Schur-convex} if
\[
\mathbf{x} \prec \mathbf{y} \quad \Rightarrow \quad F(\mathbf{x}) \leq F(\mathbf{y}).
\]
\end{definition}
Schur-convex functions favor vectors that are more spread out, making them useful for bounding symmetric functionals.


\subhead{Schur--Ostrowski Criterion \cite{peajcariaac1992convex}}
Let \( f: \mathbb{R}^d \to \mathbb{R} \) be a symmetric function with continuously differentiable partial derivatives. Then \( f \) is Schur-convex if and only if for all \( \mathbf{x} \in \mathbb{R}^d \) and for all \( 1 \leq i, j \leq d \), the following inequality holds: 
\[
(x_i - x_j)\left( \frac{\partial f}{\partial x_i} - \frac{\partial f}{\partial x_j} \right) \geq 0 
.  
\]
The following result is proven in Appendix~\ref{mafschur}.

\begin{figure*}[htbp]
    \centering
    \begin{subfigure}[b]{0.24\textwidth}
        \includegraphics[width=\linewidth]{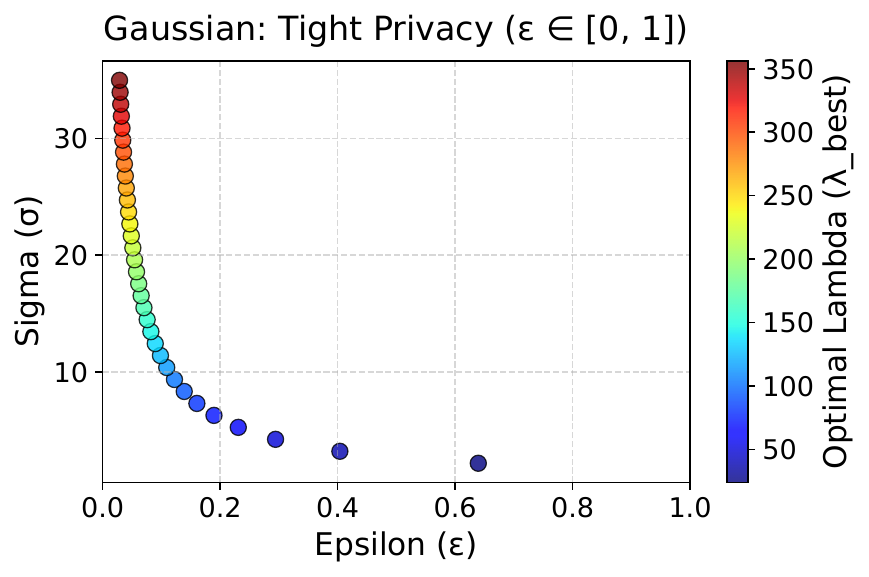}
        \caption{Gaussian Tight Privacy}
        \label{fig:fig3}
    \end{subfigure}
    \hfill
    \begin{subfigure}[b]{0.24\textwidth}
        \includegraphics[width=\linewidth]{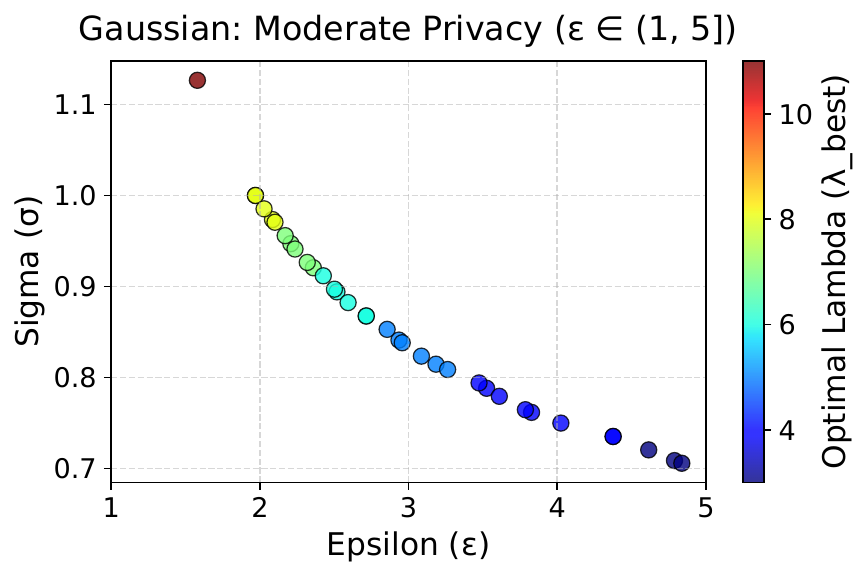}
        \captionsetup{font=small}
        \caption{Gaussian Moderate Privacy}
        \label{fig:fig4}
    \end{subfigure}
    \hfill
    \begin{subfigure}[b]{0.24\textwidth}
        \includegraphics[width=\linewidth]{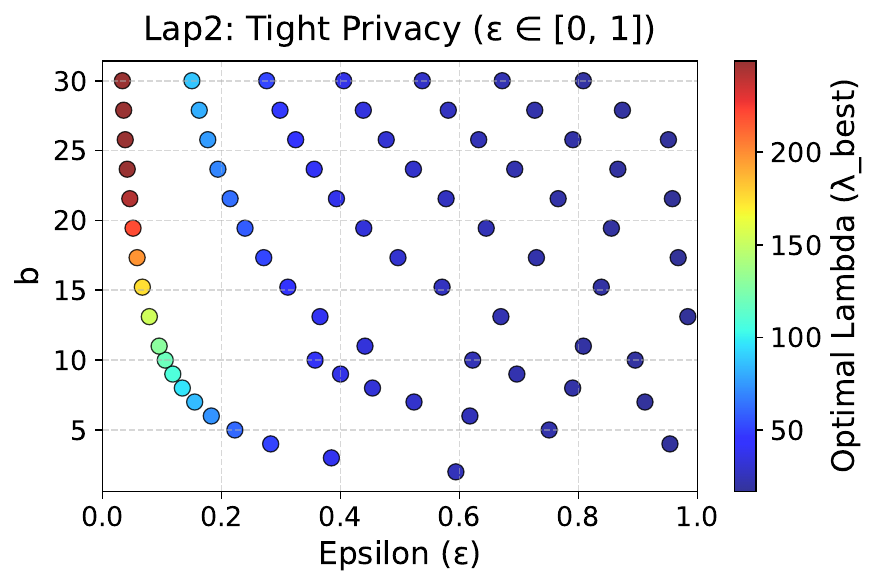}
        \caption{Lap2 Tight Privacy}
        \label{fig:fig1}
    \end{subfigure}
    \hfill
    \begin{subfigure}[b]{0.24\textwidth}
        \includegraphics[width=\linewidth]{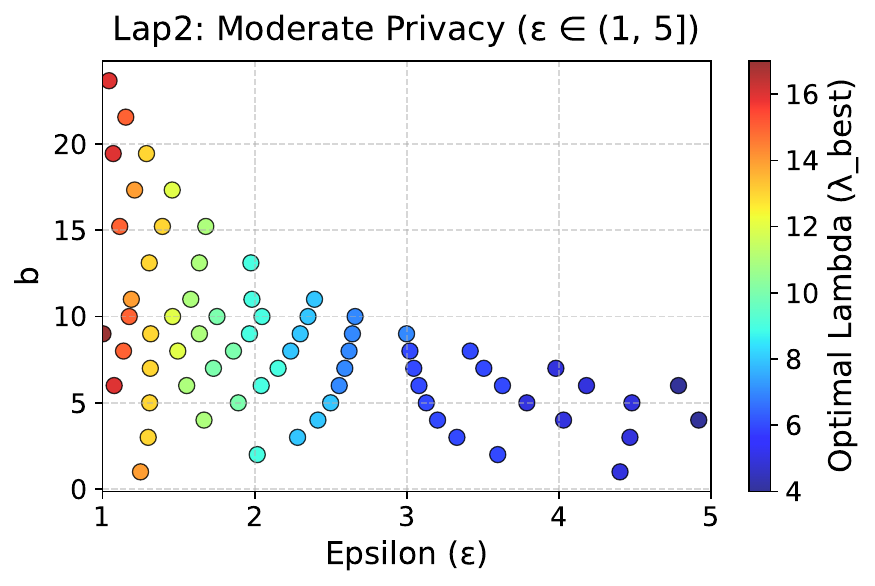}
        \caption{Lap2 Moderate Privacy}
        \label{fig:fig2}
    \end{subfigure}

\caption{In Gaussian DP-SGD, the so-called \textit{privacy wall} leads to sharp increases in the noise multiplier under tight privacy ($\epsilon < 0.5$). This arises because the injected noise scales with $\sigma \cdot C$, which becomes problematic in settings like pretraining where large $C$ is needed to preserve gradient signal—resulting in excessive noise. In contrast, \textsc{Lap}\textsubscript{2} supports larger noise or larger clipping norms (see Appendix~\ref{sec:evo_lap2}), due to its heavy tails and efficient high-order moments accounting.}
\vspace{-0.2in}\label{fig:all_four}
\end{figure*}

\begin{theorem}[Schur-Convexity of MAF]
\label{schurmaf}
Let $\bar{\mathbf{G}} = \{ \bar{g}_1, \dots, \bar{g}_n \}$ denote a single clipped gradient vector of dimension $n$, where each coordinate $\bar{g}_i$ is obtained by applying $\ell_2$ clipping to the original gradient followed by independent Laplace noise. Let $\alpha_{\bar{g}_i}(\lambda)$ denote the moments accountant function for coordinate $i$.

Then, the total moments accountant function
\begin{equation}
\alpha_{\bar{\mathbf{G}}}(\lambda) = \sum_{i=1}^n \alpha_{\bar{g}_i}(\lambda)
\end{equation}
is \textbf{Schur-convex} with respect to the vector of unsigned marginal clipped gradients, i.e., the coordinate-wise magnitudes $|\bar{g}_1|, \dots, |\bar{g}_n|$.
\end{theorem}

With $\alpha(\lambda)$ being Schur-convex, we now introduce a majorization set over $\mathbf{G}$, the $\ell_2$ clipped (but not noisy) gradient vector. The following result is proven in Appendix~\ref{majorseet}.

\begin{lemma}[Majorization Set Construction]
\label{lem:majorset}
Let $\bar{\mathbf{G}} = \{ \bar{g}_1, \dots, \bar{g}_n \}$ denote a single clipped gradient vector of dimension $n$, where each coordinate $\bar{g}_i$ is obtained by applying $\ell_2$ clipping to the original gradient. Then, $\mathbf{G}$ is weakly majorized by the vector $x = \{x_1, \dots, x_n\}$ defined by
\[
x_i = C \left( \sqrt{i} - \sqrt{i-1} \right), \quad i = 1, \dots, n,
\]
where $C$ is the $\ell_2$ clipping threshold. That is,
\(
\mathbf{G} \prec_w x,
\).
\end{lemma}

We now extend the uni-variate moments accountant to the multivariate setting by summing across coordinates, where each coordinate $i$ uses its corresponding majorization bound $x_i$. The following theorem naturally follows from Theorems~\ref{MAflaplace}, \ref{schurmaf}, and Lemma~\ref{lem:majorset}.

\begin{theorem}[Subsampled Multi-variate Laplace Mechanism]
\label{MAflaplace-multi}
Let $\bar{\mathbf{G}} = \{ \mathbf{\bar{g}}_1, \dots, \mathbf{\bar{g}}_n \}$ be the set of differentially private gradients obtained by applying $C$-$\ell_2$ clipping followed by a multi-variate Laplace mechanism with scale $b$ and sampling probability $\zeta$. Define the majorization set $x = \{x_1, \dots, x_n\}$ by $x_i = C \left( \sqrt{i} - \sqrt{i-1} \right), \quad i = 1, \dots, n$. Then, the total moments accountant satisfies
\begin{equation}
\label{mafpara}
\alpha_{\bar{\mathbf{G}}}(\lambda) \leq \sum_{i=1}^n \log \left[
 \sum_{\eta=0}^{\lambda+1} \binom{\lambda+1}{\eta} (1-\zeta)^{\lambda+1-\eta} \zeta^\eta F(x_i, \eta)
\right],
\end{equation}
where
\begin{equation}
\label{eqn:G-multi}
F(x_i, \eta) = \frac{\eta e^{\frac{(\eta-1)x_i}{b}} + (\eta-1) e^{-\frac{\eta x_i}{b}}}{2\eta-1}.
\end{equation}
\end{theorem}

\section{Analysis and Framework}
\label{sec:analysis}
In this section, we provide a theoretical analysis of \textsc{Lap}\textsubscript{2}, including per-round privacy accounting, utility measured by the signal-to-noise ratio, and insights into parameter selection.

\begin{figure*}[t]
    \centering
    \includegraphics[width=0.8\linewidth, trim=20 0 30 0, clip]{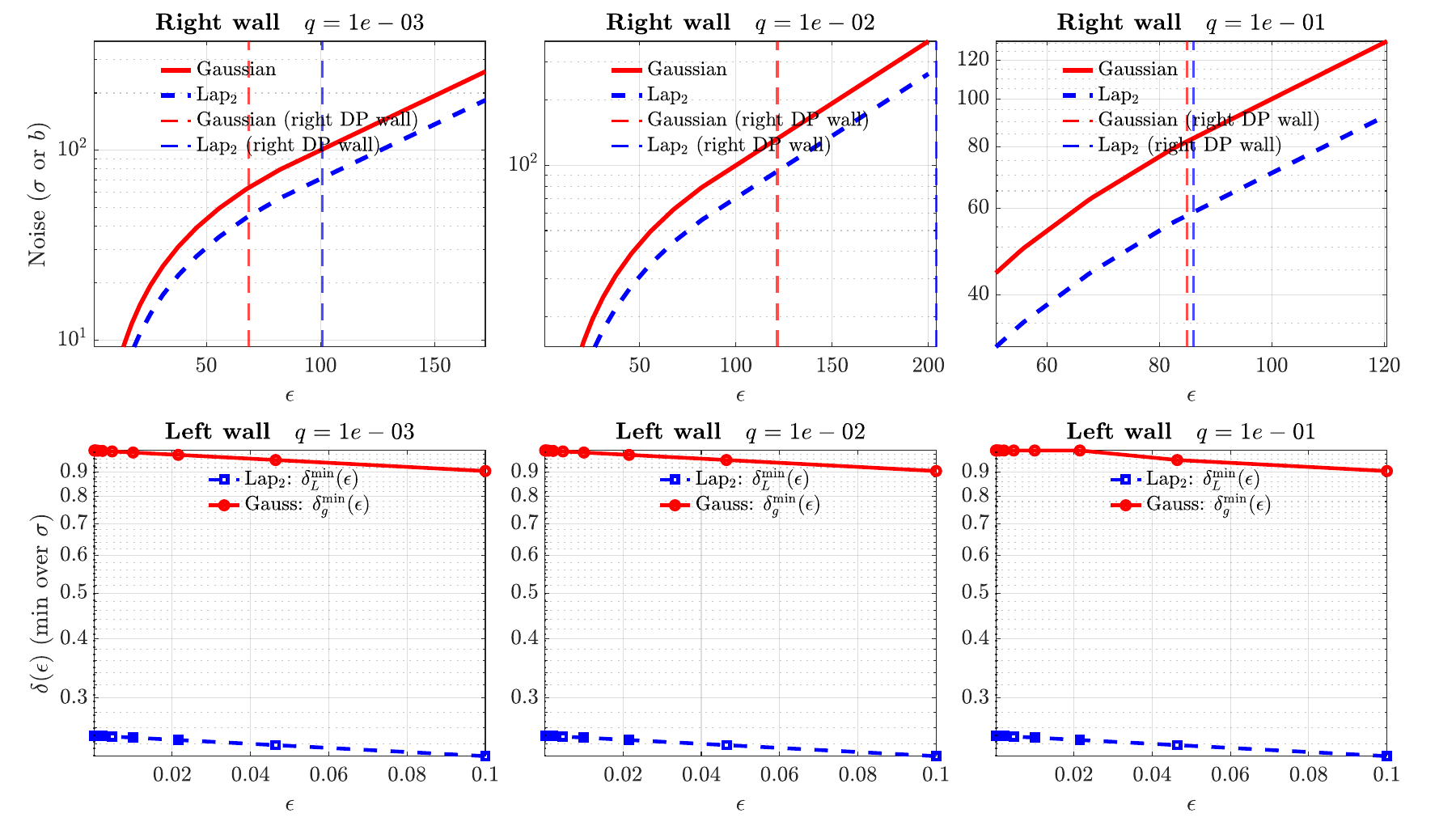}
    \vspace{-0.1in}\caption{\textbf{Illustration of two-sided privacy walls} (left and right walls) during CNN training on the MNIST dataset under DP-SGD. 
    The \emph{left wall} corresponds to the high-privacy regime, where the noise scale saturates and $\delta(\epsilon)$ approaches 1, indicating limited further privacy gain.
    The \emph{right wall} corresponds to the low-privacy regime, where the effective signal-to-noise ratio no longer improves with larger $\epsilon$.
    Together, these walls define the practical operating range (privacy corridor) in which training remains both private and useful.}
    \vspace{-0.15in}
    \label{Fig:Privacywalls}
\end{figure*}

   


\subsection{Privacy Accounting: \textsc{Lap}\textsubscript{2} vs Gaussian}

Analyzing Equations~\ref{A_alpha} and~\ref{mafpara}, the moments accountant functions (MAFs) for the subsampled \textsc{Lap}\textsubscript{2} and Gaussian mechanisms, yields two critical observations:

\textbf{(1) Lap\textsubscript{2} mitigates the privacy wall.}  
Differentially private training is fundamentally constrained by two characteristic limits, termed \emph{two-sided privacy walls}.  
The \textbf{left wall}, which primarily affects the Gaussian mechanism, emerges in the high-privacy regime ($\epsilon \!\to\! 0$) where $\delta(\epsilon)$ approaches~1. In this regime, increasing the noise scale yields little to no additional privacy gain, indicating the onset of \emph{privacy saturation}.  
While both Gaussian and Laplace mechanisms exhibit this saturation behavior, the Gaussian variant suffers more severely due to the rapid rise in its DP failure term $\delta$, which can render its guarantees vacuous. The \textbf{right wall} arises in the low-privacy regime ($\epsilon \!\gg\! 1$), where the effective signal-to-noise ratio (SNR) no longer improves with larger~$\epsilon$, marking the onset of \emph{utility saturation}.  
Between these limits lies a narrow \emph{privacy corridor} in which both privacy and accuracy remain meaningful.

Unlike the Gaussian mechanism, \textsc{Lap}\textsubscript{2} demonstrates strong resistance to the left privacy wall, as seen in Figure~\ref{fig:fig3} and \ref{fig:fig4} and previously noted in~\cite{mironov2019r,TAN}. This stability arises from the heavier tails of Laplace noise (Equation~\ref{eqn:G-multi}), which allow more efficient privacy budget consumption in both high- and low-privacy regimes. Figure~\ref{fig:fig2} and \ref{fig:fig2} further shows this advantage over different clipping values and their optimal moment orders~$\lambda^\star$, where smaller clipping values (dominant in high-privacy settings) correspond to larger exploitable moments.\footnote{\textsc{Lap}\textsubscript{2} results are based on a 26K-parameter CNN with $\delta = 10^{-5}$.}

Figure~\ref{Fig:Privacywalls} summarizes these phenomena for CNN training on MNIST under Gaussian and \textsc{Lap}\textsubscript{2} DP-SGD at subsampling rates $q\!\in\!\{10^{-3},10^{-2},10^{-1}\}$. In the \textbf{right-wall regime} (top row), the slope $W_R(\epsilon)=|d\log\sigma/d\log\epsilon|$ deviates from the ideal $1/\epsilon$ scaling, signaling diminishing utility returns. In the \textbf{left-wall regime} (bottom row), $\delta_g(\epsilon)$ increases faster than $\delta_{L2}(\epsilon)$, and the first $\epsilon$ satisfying $\delta_g > 2\delta_{L2}$ marks the onset of privacy saturation.  
As $q$ increases, both walls move inward, narrowing the privacy corridor.  
Overall, \textsc{Lap}\textsubscript{2} consistently delays both boundaries relative to Gaussian DP-SGD, expanding the stable region of privacy-utility trade-offs.

A practical limitation remains for large-scale models: summing moment-aggregated functions (MAFs) across millions of coordinates, as required by majorization (Equation~\ref{mafpara}), inflates the overall privacy bound.  
Hence, \textsc{Lap}\textsubscript{2} is most effective for small- to medium-scale architectures such as MobileNet and compact ResNet variants (e.g., ResNet-18), where the two walls remain well-separated and its benefits are most pronounced.



\textbf{(2) Clipping-aware accounting.}
Traditional DP-SGD decouples clipping and noise addition—$\sigma$ governs the privacy guarantee, while $C$ only scales the added noise. In contrast, our \textsc{Lap}\textsubscript{2}-based accounting (Theorem~\ref{subsampledLap2}) integrates $C$ directly into the moments accountant, resulting in tighter privacy bounds and more favorable inflection points.

Figure~\ref{fig:acc_eps_clip_b} empirically supports this claim: with accountant-adjusted $b$, larger $C$ values yield over 90\% accuracy at moderate privacy budgets (e.g., $\epsilon = 5$). Under tighter budgets (e.g., $\epsilon = 1$), smaller $C$ is required, leading to reduced utility. Furthermore, we performed a grid search over the configurations $(q,C,b)$ to identify the optimal setting for $\epsilon=0.5, 1, 2$, respectively, and $\delta=10^{-5}$. The results of this search are presented in Figure~\ref{fig:3dplots}, which shows the accuracy of the CNN trained on the MNIST dataset. The optimal configuartion (from optimizer) stays within the optimal - yellow region and is highlighted using star mark in the figure. 

\begin{figure}[!h]
\centering
\includegraphics[width=0.85\linewidth]{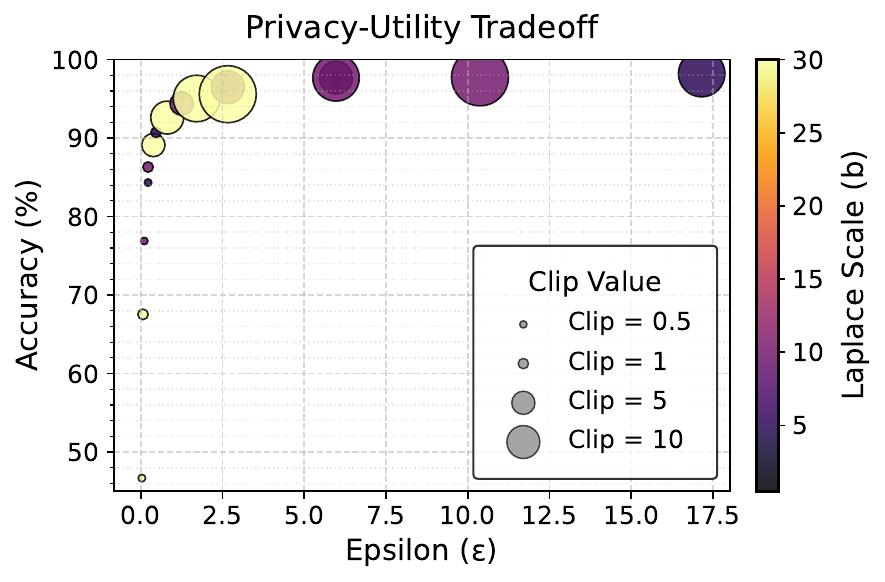}\vspace{-0.1in}
\caption{
Accuracy vs. privacy budget $\epsilon$ for different clipping norms $C$ using \textsc{Lap}\textsubscript{2} (CNN on MNIST, $\delta=10^{-5}$). 
Increasing $C$ improves utility under moderate $\epsilon$ due to stronger signal retention, but high $C$ becomes suboptimal for very tight $\epsilon$ as noise grows superlinearly.
\vspace{-0.2in}}
\label{fig:acc_eps_clip_b}
\end{figure}

\begin{figure*}[t]
  \centering
  \begin{subfigure}[b]{0.32\textwidth}
    \centering
    \includegraphics[width=\linewidth]{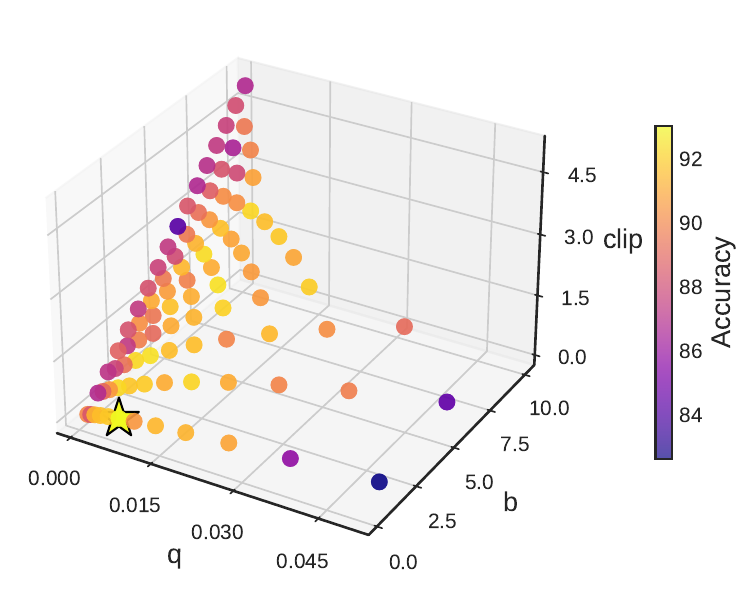}
    \caption{$\epsilon = 0.5$}
    \label{fig:3d_1}
  \end{subfigure}
  \hfill
  \begin{subfigure}[b]{0.32\textwidth}
    \centering
    \includegraphics[width=\linewidth]{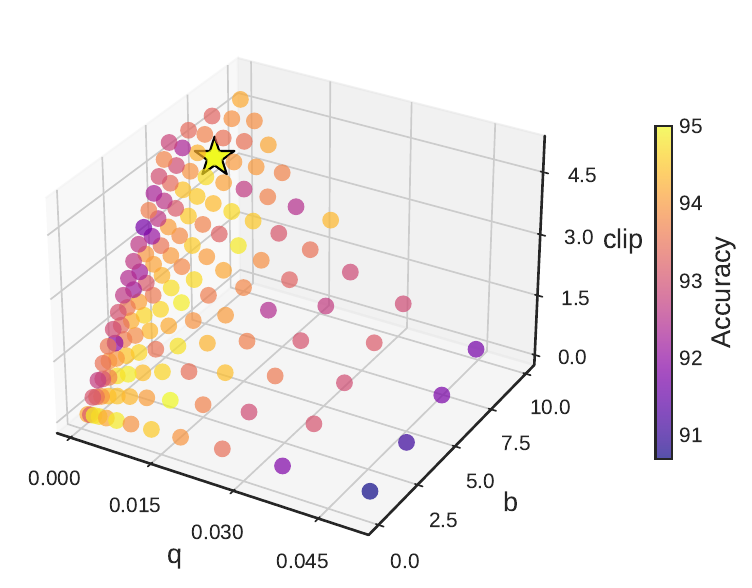}
    \caption{$\epsilon = 1$}
    \label{fig:3d_2}
  \end{subfigure}
  \hfill
  \begin{subfigure}[b]{0.32\textwidth}
    \centering
    \includegraphics[width=\linewidth]{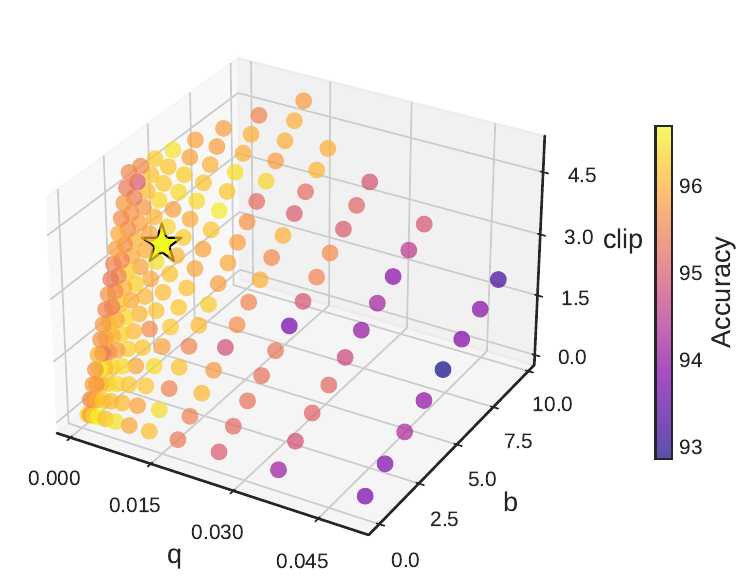}
    \caption{$\epsilon = 2$}
    \label{fig:3d_3}
  \end{subfigure}

  \caption{
  Accuracy of a CNN trained on MNIST under DP-SGD with \textsc{Lap}\textsubscript{2} for $20$ epochs.
  Configurations $(q, b, c)$ were selected via grid search for each target privacy budget $\epsilon \in \{0.5, 1, 2\}$ and $\delta=10^{-5}$,
  with color indicating the resulting accuracy.
  The star marks the configuration achieving maximum accuracy for each $\epsilon$.
  } \vspace{-0.2in}
  \label{fig:3dplots}
\end{figure*}

\vspace{-0.1in}
\subsection{Optimal Parameter Selection}
\label{sec:opt}

Let each individual gradient be clipped to norm \( C \) with the sampling rate \( \zeta = B/N \), batch size \( B \) and dataset size \( N \). In a single step, the expected signal contribution per example is \( C / B \) with probability \( \zeta \), and zero otherwise. Over \( T \) steps, the accumulated signal per sample is thus \( \zeta T C / B = T C / N \), and its squared norm scales as \( T^2 C^2 / N^2 \). 

In Gaussian DP-SGD, the total noise added per step is \( \sigma C \), so the noise variance over \( T \) steps is \( T C^2 \sigma^2 / B^2 \). This motivates the signal-to-noise ratio (SNR)
\[
\eta^2 = \frac{T^2 C^2 / N^2}{T C^2 \sigma^2 / B^2}
= \frac{T C^2}{N^2} \cdot \frac{B^2}{C^2 \sigma^2}
= \frac{T B^2}{N^2 \sigma^2}
= \frac{T \zeta^2}{\sigma^2}.
\]
This is the key driver of utility–privacy tradeoff: as shown in~\cite{TAN}, the RDP privacy loss \( \epsilon \) under Gaussian DP-SGD can be approximated as
\[
\epsilon_{RDP} \approx \eta^2 + 2\eta \sqrt{\log(1/\delta)}.
\]

We refer to \( \eta = \zeta / \sigma \) as the \emph{per-step SNR}, and \( \eta^2 = T \zeta^2 / \sigma^2 \) as the \emph{total SNR} across \( T \) steps.

\vspace{0.5em}
\noindent\textbf{Extension to \textsc{Lap}\textsubscript{2}.} Unlike the Gaussian mechanism, where \( \sigma \) and \( C \) appear as a multiplicative pair in the total noise, the multivariate Laplace mechanism used in \textsc{Lap}\textsubscript{2} introduces noise drawn from a symmetric Laplace distribution with dispersion \( b \), while gradient sensitivity is controlled by clipping to \( C \). In this setting, privacy loss depends on the ratio $\rho = \frac{C}{b}$, which we refer to as the \emph{privacy-relevant sensitivity-to-noise ratio}. Figure \ref{fig:acc_eps_clip_b} confirms that increasing $\rho = C/b$ (subject to DP constraints) yields consistent accuracy gains. $\rho$ also governs the growth of moment terms in the privacy accountant. Specifically, in the small-\( \zeta \), small-\( \rho \) regime, the moments accountant satisfies
\[
\alpha(\lambda) \approx T \zeta^2 \kappa_\lambda \rho^2, \quad \text{where} \quad \kappa_\lambda = \frac{\lambda(\lambda+1)}{2(2\lambda+1)}.
\]

Plugging into the MAF\(\to\)DP conversion,
\[
\epsilon(\delta)
\approx \min_{\lambda > 0} \left(
\frac{T \zeta^2 \kappa_\lambda \rho^2}{\lambda}
- \frac{\log \delta}{\lambda}
\right),
\]
which yields an optimal \( \rho^* \) under privacy target \( \epsilon_{\text{tar}} \) and failure probability \( \delta \):
\[
\rho^* \approx \frac{\epsilon_{\text{tar}}}{2 \zeta} \sqrt{\frac{1}{T \log(1/\delta)}}.
\]

From this, we obtain a first-order estimate of the optimal noise parameter \( b \) in terms of a chosen clipping norm \( C \):
\[
b^* \approx \frac{2 \zeta}{\epsilon_{\text{tar}}} \sqrt{T \log(1/\delta)} \; C.
\]

This formula gives a closed-form initialization for \((C, b)\), which can then be refined using exact privacy computation (e.g., Algorithm~\ref{Hfunction}). In practice, we recommend choosing \( C \) based on utility constraints (e.g., a percentile of per-example gradient norms), and computing \( b \) accordingly using the expression above. 
\subsection{The Lap\textsubscript{2} Framework}
\label{sec:framework}

The Lap\textsubscript{2} framework enables practical $(\epsilon, \delta)$-DP training under \textit{$\ell_2$ clipping} with \textit{Laplace noise}. To balance privacy and utility, Lap$_2$ optimizes the clipping threshold $C$ under a fixed privacy budget $(\epsilon, \delta)$. Unlike in Gaussian DP-SGD where clipping is decoupled from noise, our framework leverages a signal-to-noise ratio (SNR) interpretation where the signal grows with $C$ and the noise is implicitly governed by $b$, the Laplace scale. Specifically, under our $\ell_2$-clipped Laplace mechanism, the per-step SNR is approximated by:
$
\eta = \frac{C}{b},
$
and total SNR across $T$ steps is $\eta^2 T$. A higher SNR implies better utility; hence, we aim to maximize $\eta$—equivalently, maximize $C$ for a fixed $b$, subject to satisfying the privacy constraint $\epsilon(C, b) \leq \epsilon_{\text{target}}$. This motivates the following optimization problem:
\[
\max_{C,\ b} \ \frac{C}{b} \quad \text{s.t.} \quad \epsilon(C, b) \leq \epsilon_{\text{target}}, \quad C, b > 0.
\]
Due to the non-linear nature of the privacy computation (based on majorization and moments accounting), we perform a grid or binary search to efficiently explore feasible $(C, b)$ pairs. While the default implementation in Algorithm~\ref{Hfunction} uses grid search over $(C, b)$ ranges, a potential enhancement is to replace the inner loop with a binary search over $b$ for each fixed $C$. This accelerates convergence towards the optimal pair $(C^*, b^*)$ satisfying the privacy constraint.

In practice, we increment $C$ in ascending order (e.g., logarithmic or linear grid), and for each candidate $C$, we search for the smallest $b$ such that the moments accountant returns $\epsilon(C, b) \leq \epsilon_{\text{target}}$. The largest $C$ satisfying this constraint yields the highest possible SNR. This process is formalized in Algorithm~\ref{Hfunction}, and the resulting $(C^*, b^*)$ pair is passed to the DP-SGD loop in Algorithm~\ref{alg:lap2} in Appendix.

\begin{algorithm}[t]
\caption{Lap$_2$ Parameter  \( b \) Optimizer (via Binary Search)}
\label{Hfunction}
\begin{footnotesize}
\begin{tabbing}
\hspace{1em} \= \hspace{1em} \= \hspace{1em} \= \kill
\textbf{Input:} rounds $T$, sampling rate $\zeta$, privacy target $(\epsilon, \delta)$,\\
\> max values $C_{\text{max}}, b_{\text{max}}$; resolution $\Delta_C$, binary tolerance $\tau$ \\
\textbf{Output:} optimal $(C^*, b^*)$ maximizing $C$ under $(\epsilon, \delta)$ \\

$C^* \leftarrow 0$, $b^* \leftarrow \texttt{None}$ \\

\textbf{for} $C$ in $[C_{\text{min}}, C_{\text{max}}]$ step $\Delta_C$ \textbf{do} \\
\> $b_{\text{low}} \leftarrow b_{\text{min}},\ b_{\text{high}} \leftarrow b_{\text{max}}$ \\
\> \textbf{while} $b_{\text{high}} - b_{\text{low}} > \tau$ \textbf{do} \\
\> \> $b \leftarrow (b_{\text{low}} + b_{\text{high}})/2$ \\
\> \> Define majorization set: $x_i = C \cdot (\sqrt{i} - \sqrt{i-1})$ for $i = 1$ to $n$ \\
\> \> Initialize total accountant $\alpha \leftarrow 0$ \\
\> \> \textbf{for} $t = 1$ to $T$ \textbf{do} \\
\> \> \> Compute $F(x_i, \eta)$ and moments $\alpha_t$ (see Equation~\ref{mafpara}) \\
\> \> \> $\alpha \leftarrow \alpha + \alpha_t$ \\
\> \> \textbf{end for} \\
\> \> Compute $\epsilon_C(b)$ from moment accountant \\
\> \> \textbf{if} $\epsilon_C(b) \leq \epsilon$ \textbf{then} \\
\> \> \> $b_{\text{high}} \leftarrow b$ \\
\> \> \textbf{else} \\
\> \> \> $b_{\text{low}} \leftarrow b$ \\
\> \textbf{end while} \\
\> \textbf{if} $b_{\text{high}} < b_{\text{max}}$ \textbf{and} $C > C^*$ \textbf{then} \\
\> \> $(C^*, b^*) \leftarrow (C, b_{\text{high}})$ \\
\textbf{end for} \\

\textbf{return} $(C^*, b^*)$
\end{tabbing}
\end{footnotesize}
\end{algorithm}

\section{Experiments}
\label{sec:exp}

In this section, we evaluate the performance of our \textsc{Lap}\textsubscript{2} mechanism in terms of privacy, accuracy, and efficiency on computer vision (CV) and natural language processing (NLP) tasks. First, we compare the utility performance with two baselines: the standard DP-SGD using Gaussian noise, and the classical Laplace mechanism with $\ell_1$-norm clipping. 

\begin{table*}[]
\centering
\scriptsize
\caption{Hyperparameter settings for different CV datasets and models.}
\renewcommand{\arraystretch}{1.2} 
\setlength{\tabcolsep}{6pt} 
\begin{tabular}{|c|c|c|c|c|c|c|c|c|c|}
\hline
\textbf{Dataset} & \textbf{Model} & \textbf{Sampling} & \textbf{Clipping} & \textbf{$\delta$} & \textbf{Learning}  & \textbf{Label} & \textbf{Weight} & \textbf{Training} \\  
 &  & \textbf{Rate}  & \textbf{Threshold} &  & \textbf{Rate}  & \textbf{Smoothing} & \textbf{Decay} & \textbf{Steps} \\ \hline
MNIST & CNN & 0.0043 & 1 & $1\times10^{-5}$ & $1\times10^{-3}$  & 0.15 & $1\times10^{-4}$ & 5860 \\ \hline
Fashion-MNIST & CNN & 0.0043 & 1 & $1\times10^{-5}$ & $1\times10^{-3}$  & 0.15 & $1\times10^{-4}$ & 5860 \\ \hline
CIFAR-10 & ViT & 0.01668 & 1 & $1\times10^{-5}$ & $1\times10^{-3}$ & 0 & $1\times10^{-4}$ & 900 \\ \hline
\end{tabular}
\label{tab:hyper_cv}
\end{table*}

\subsection{Experimental Settings}

\noindent \textbf{Computer Vision Datasets and Tasks}. Image classification experiments were performed using three standard datasets: CIFAR-10, MNIST, and Fashion-MNIST. The CIFAR-10 dataset contains 60,000 32$\times$32 color images across 10 classes (6{,}000 per class), divided into 50,000 training and 10{,}000 test images. The MNIST dataset consists of 70,000 28$\times$28 grayscale images of handwritten digits ($0-9$), with 60{,}000 for training and 10{,}000 for testing. Similarly, the Fashion-MNIST dataset includes 70{,}000 28$\times$28 grayscale images of 10 clothing categories. The full dataset was used for evaluation.


\vspace{0.05in}

\noindent
\textbf{CV Models}.
Testing with MNIST and Fashion MNIST was done on a small 4-layer CNN model as described in the \href{https://github.com/tensorflow/privacy}{Tensorflow privacy tutorial} for CV tasks. 
    \noindent
Testing with CIFAR-10 was done with a ViT \cite{dosovitskiy2021imageworth16x16words} model with a patch size of 16, as sourced from the \href{https://timm.fast.ai/}{TIMM Repository}. The images were scaled to the expected input size of 224x224 pixels for the ViT model. No other augmentations were applied to the images.

\vspace{0.05in}

\noindent 
\textbf{Natural Language Processing Datasets and Tasks}.
We first evaluate the \textsc{Lap}\textsubscript{2} on the sentiment analysis tasks from the GLUE benchmark \cite{wang2018glue}. Specifically, following \cite{yu2021differentially}, we fully fine-tune the RoBERTa-base model on SST-2 and QNLI datasets \cite{wang2018glue}. 
These datasets are widely used to evaluate private training. 
SST-2 has more than 60k samples in the training set and 1,821 samples in the test set; QNLI has more than 100k samples in the training set and 5,463 samples in test, including two classes. 
We also evaluate the table-to-text generation task's performance that generates the table entries' descriptions. We fine-tune the DistilGPT2 model \cite{yu2021differentially} with the E2E dataset \cite{novikova2017e2e}.


\vspace{0.05in}

\noindent \textbf{Experimental Platform}. Experiments were conducted on three systems optimized for specific tasks: (1) a high-end workstation with an AMD Ryzen Threadripper PRO 5975WX (32 cores, 64 threads), 500\,GB RAM, and 3$\times$~NVIDIA Quadro RTX A6000 (48\,GB) GPUs for NLP tasks; (2) a virtualized server with up to 192 CPU cores (Intel Xeon Platinum and AMD EPYC), 1\,TB RAM, used primarily for ViT training; and (3) a consumer-grade system with an AMD Ryzen 7 8700F (8 cores, 16 threads), 32\,GB RAM, and an NVIDIA RTX 3060 (12\,GB VRAM) for CNN and ViT experiments on MNIST, Fashion-MNIST, and CIFAR-10.

\subsection{Utility Evaluation}

\begin{table}[h!]
\centering
\caption{Comparison of Gaussian, Laplace ($\ell_1$ norm), and \textsc{Lap}\textsubscript{2} ($\ell_2$ norm) mechanisms across different $\epsilon$.}
\begin{tabular}{|c|c|c|c|}
\hline
\textbf{$\epsilon$} & \textbf{Gaussian (\%)} & \textbf{Laplace (\%)} & \textbf{\textsc{Lap}\textsubscript{2} (\%)} \\
\hline
\multicolumn{4}{|c|}{\textbf{CNN on MNIST}} \\
\hline
3.42 & 97.27 & 51.82 & 96.82 \\
2.53 & 97.19 & 46.67 & 96.90 \\
1.68 & 97.42 & 28.82 & 95.39 \\
0.88 & 96.08 & 16.44 & 93.29 \\
0.13 & 87.44 & 10.40 & 78.96 \\
\hline
\multicolumn{4}{|c|}{\textbf{CNN on FMNIST}} \\
\hline
3.42 & 83.67 & 55.05 & 82.60 \\
2.53 & 83.09 & 43.85 & 81.45 \\
1.68 & 82.34 & 43.94 & 81.45 \\
0.88 & 80.34 & 28.22 & 77.03 \\
0.13 & 72.58 & 14.74 & 70.03 \\
\hline
\multicolumn{4}{|c|}{\textbf{ViT on CIFAR-10 (Fine-tuning)}} \\
\hline

0.75 & 97.17 & 53.58 & 98.11 \\
0.5  & 96.90 & 47.04 & 98.18 \\

\hline
\end{tabular}
\label{tab:privacy_comparison_all}
\end{table}

\noindent \textbf{Utility of CV Tasks}. To demonstrate the efficacy of the proposed \textsc{Lap}\textsubscript{2} framework, we first discuss results of classification on the CV datasets (MNIST and Fashion-MNIST datasets) using a small 4-layer CNN, shown in Table \ref{tab:privacy_comparison_all}. The privacy budget is in the range $0.1$ to $3.5$. Other training details is shown in the Table \ref{tab:hyper_cv}. Specifically, we compare the test accuracy achieved using the Gaussian mechanism, the Laplace mechanism with the $\ell_1$ norm, and the \textsc{Lap}\textsubscript{2} mechanism. The first key trend observed is that test accuracy increases with larger values of $\epsilon$, reflecting the fundamental trade-off between privacy and utility: a higher $\epsilon$ relaxes the privacy constraint, allowing less noise to be added and thereby improving model performance. The Gaussian mechanism performs well across the board, achieving over 96\% accuracy on MNIST when $\epsilon \geq 0.88$, and over 80\% on FMNIST. However, the Laplace mechanism with the $\ell_1$ norm consistently yields poor results. For instance, on MNIST with $\epsilon = 0.88$, it only achieves $16.44\% $ accuracy far below the Gaussian ($96.08\%$) and \textsc{Lap}\textsubscript{2} (93.29\%). On FMNIST, the performance gap is similarly large. This degradation is primarily due to the substantially larger $\ell_1$ norm of the gradient vector compared to its $\ell_2$ counterpart, which results in gradient loss during clipping during training. In contrast, the \textsc{Lap}\textsubscript{2} mechanism consistently matches or closely follows the performance of the Gaussian mechanism across all $\epsilon$ values. On MNIST, its accuracy remains above 93\% for $\epsilon \geq 0.88$, and still reaches 78.96\% at the extremely strict privacy level of $\epsilon = 0.13$, a $68.6$ percentage point gain over Laplace with the $\ell_1$ norm.

Furthermore, we fine-tune the ViT model on the CIFAR-10 dataset instead of training from scratch. The last section in Table \ref{tab:privacy_comparison_all} presents the fine-tuning results of the ViT model on the CIFAR-10 dataset under different noise mechanisms. 
For the privacy bound $\epsilon = 0.75$, the Gaussian and Laplace ($\ell_1$-clipping) mechanisms achieve accuracies of 97.17\% and 53.58\%, respectively, 
while \textsc{Lap}\textsubscript{2} attains a superior accuracy of 98.11\%. 
Similarly, for $\epsilon = 0.5$, the Gaussian and Laplace ($\ell_1$) methods yield 96.90\% and 47.04\%, whereas \textsc{Lap}\textsubscript{2} achieves 98.18\%. 
These results demonstrate that the proposed \textsc{Lap}\textsubscript{2} mechanism consistently outperforms both Gaussian and Laplace ($\ell_1$) approaches across different privacy budgets, 
highlighting its robustness and effectiveness for fine-tuning ViT models under differential privacy guarantees.


\color{black}

\begin{table}[h!]
\centering
\scriptsize
\caption{Comparison of Gaussian, Laplace ($\ell_1$ norm), and \textsc{Lap}\textsubscript{2} ($\ell_2$ norm) mechanisms across different privacy levels ($\epsilon$) for \texttt{RoBERTa-base} on SST-2 and QNLI.}
\begin{tabular}{|c|c|c|c|}
\hline
\textbf{$\epsilon$} & \textbf{Gaussian (\%)} & \textbf{Laplace (\%)} & \textbf{\textsc{Lap}\textsubscript{2} (\%)} \\
\hline
\multicolumn{4}{|c|}{\textbf{RoBERTa-base on SST-2}} \\
\hline
3.7382 & 90.31 & 50.34 & 90.11 \\
1.1584 & 89.65 & 50.21 & 83.32 \\
0.9108 & 89.23 & 50.15 & 89.14 \\
0.5384 & 87.16 & 48.97 & 87.88 \\
\hline
\multicolumn{4}{|c|}{\textbf{RoBERTa-base on QNLI}} \\
\hline
3.6452 & 83.26 & 50.76 & 83.73 \\
1.1365 & 82.61 & 50.49 & 82.52 \\
0.9345 & 82.17 & 50.18 & 82.87 \\
0.5168 & 80.41 & 49.97 & 60.87 \\
\hline
\end{tabular}
\label{tab:privacy_comparison_roberta_nlp}
\end{table}


\noindent \textbf{Utility of Sentiment Analysis (NLP)}. 
We conducted fine-tuning experiments for sentiment analysis using the SST-2 and QNLI datasets with the \texttt{RoBERTa-base} model. Table~\ref{tab:privacy_comparison_roberta_nlp} presents the model accuracy under various privacy levels ($\epsilon \in {0.5, 0.9, 1.1, 3.7}$), comparing our proposed \textsc{Lap}\textsubscript{2} mechanism with the Gaussian mechanism and the standard Laplace mechanism using the $\ell_1$ norm clipping.
As observed, the Laplace mechanism with the $\ell_1$ norm consistently yields accuracy around 50\% across all $\epsilon$ values on both datasets. This is close to the baseline performance of the pretrained model without fine-tuning, suggesting that the injected noise is so severe that the model fails to learn from the downstream task, rendering fine-tuning ineffective.

In contrast, our proposed \textsc{Lap}\textsubscript{2} mechanism consistently matches or outperforms the Gaussian mechanism across all $\epsilon$ values and both datasets (similar to our findings in the computer vision tasks). For the SST-2 dataset, \textsc{Lap}\textsubscript{2} achieves similar result as Gaussian when $\epsilon$ is large ($\epsilon = 3.7$). For the QNLI dataset, the trend is even more pronounced: the gap grows from roughly 0.5\% at $\epsilon = 3.6$. These results clearly demonstrate that \textsc{Lap}\textsubscript{2} maintains strong utility under strong privacy guarantees with Laplace noise, especially in the context of large pre-trained language models. 



\begin{figure}
    \centering
    \includegraphics[width=0.8\linewidth]{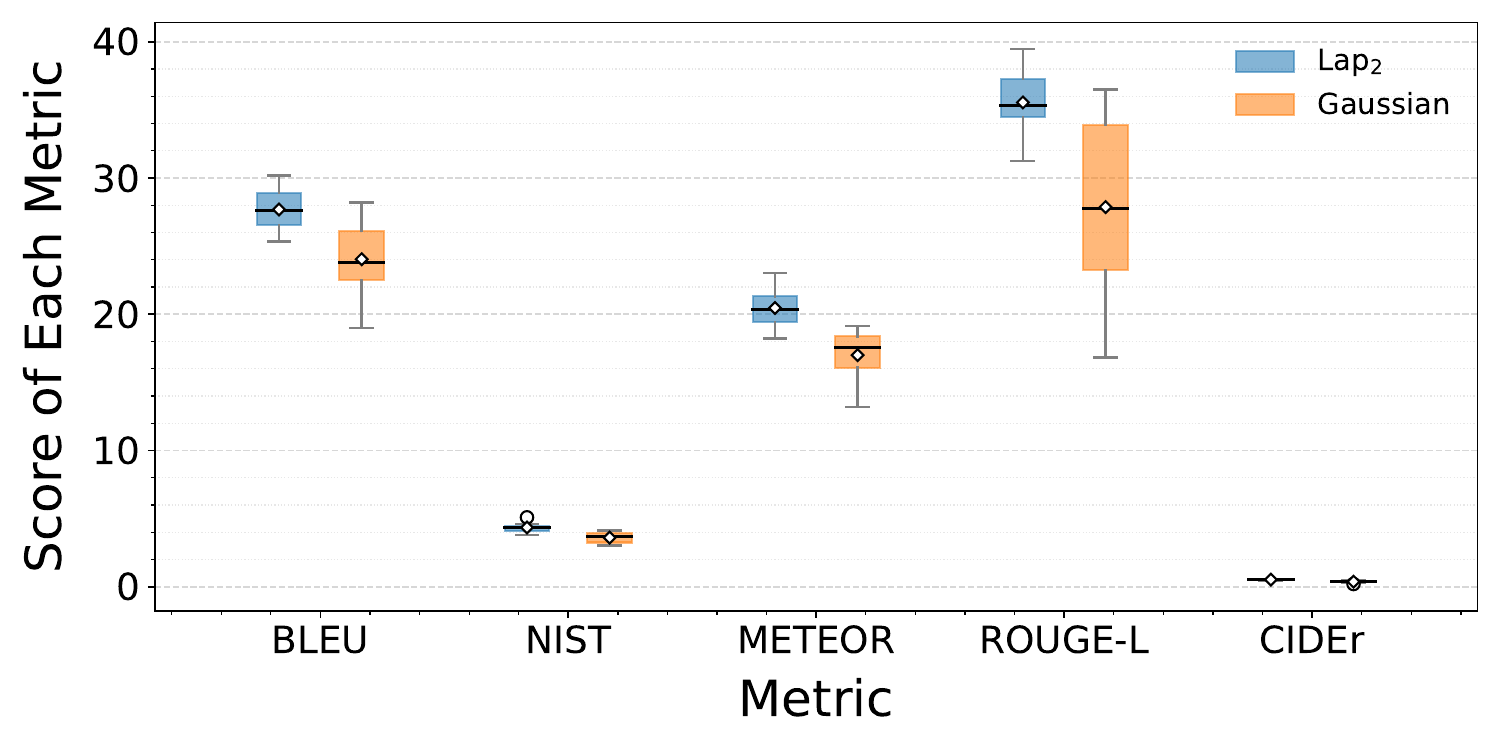}
    \caption{Gaussian and \textsc{Lap}\textsubscript{2} mechanisms on the DistilGPT-2 model and E2E dataset for the generation task (batch size $B=80$, clipping value $C=2$, and $\epsilon=1$).}\vspace{-0.2in}
    \label{fig:boxplot}
\end{figure}

\vspace{0.05in}

\begin{figure*}[htbp]
    \centering
    \begin{subfigure}[b]{0.24\textwidth}
        \includegraphics[width=\textwidth]{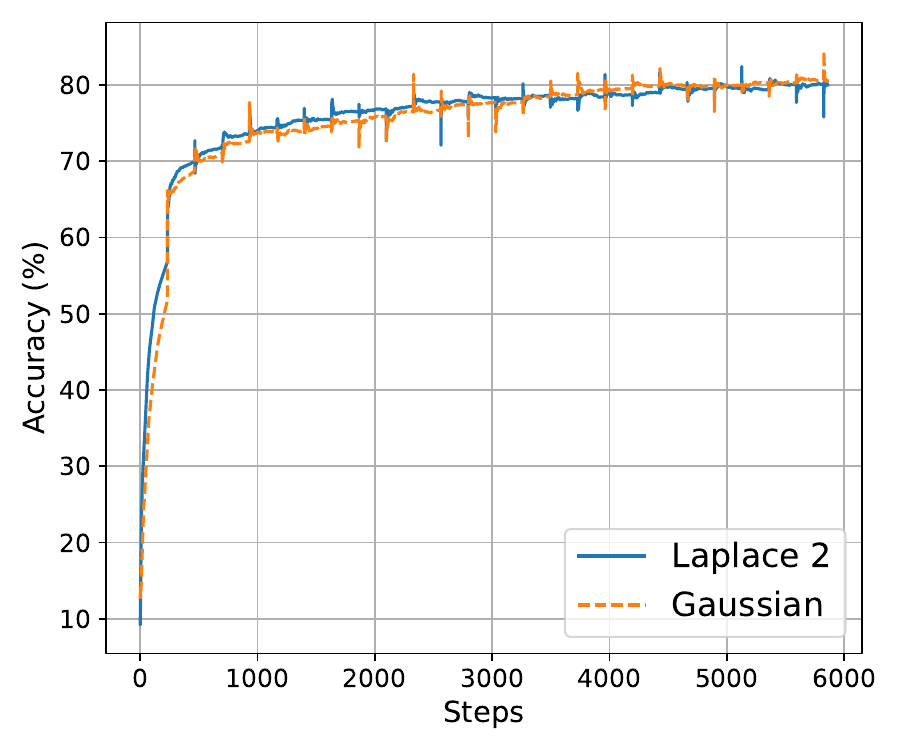}
        \caption{FMNIST, fix accuracy 80\%}
        \label{fig:runtime_1}
    \end{subfigure}
    \begin{subfigure}[b]{0.24\textwidth}
        \includegraphics[width=\textwidth]{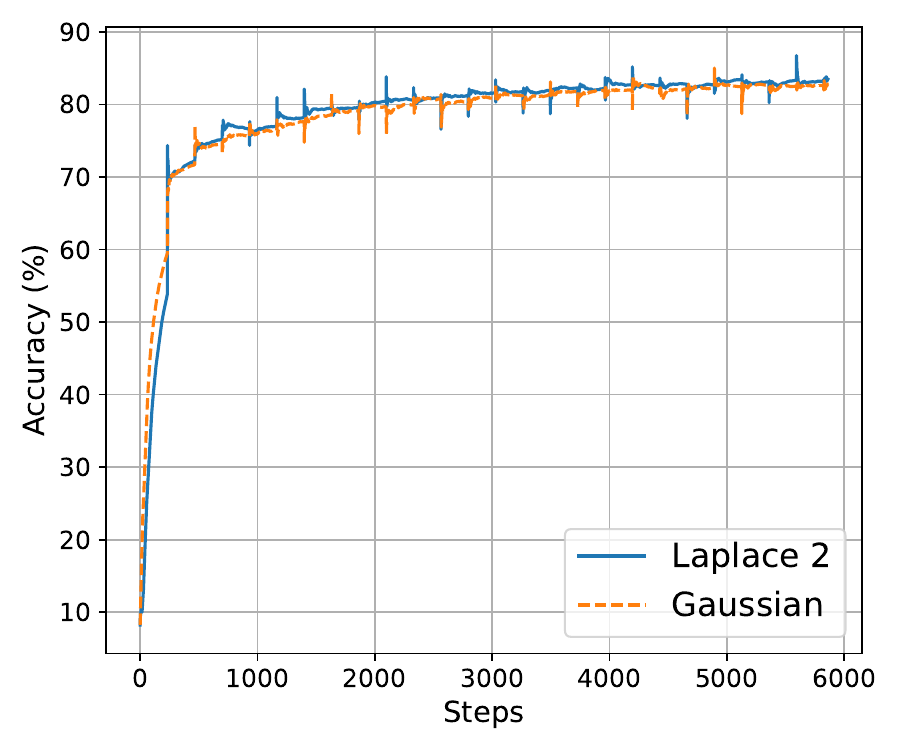}
        \caption{FMNIST, fix accuracy 82\%}
        \label{fig:runtime_2}
    \end{subfigure}
    \begin{subfigure}[b]{0.24\textwidth}
        \includegraphics[width=\textwidth]{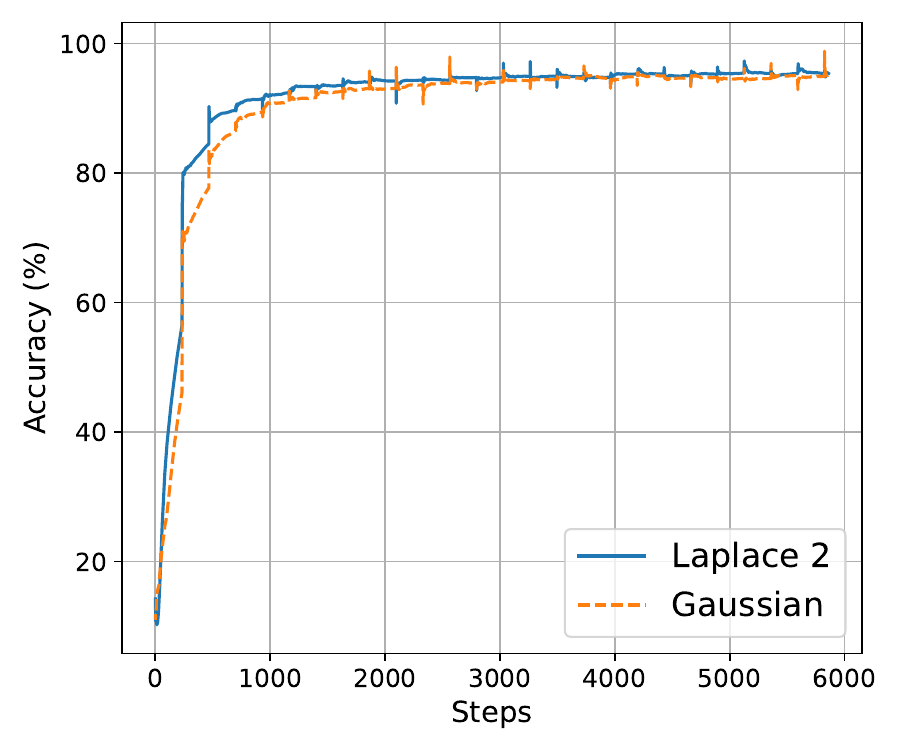}
        \caption{MNIST, fix accuracy 95\%}
        \label{fig:runtime_3}
    \end{subfigure}
    \begin{subfigure}[b]{0.24\textwidth}
        \includegraphics[width=\textwidth]{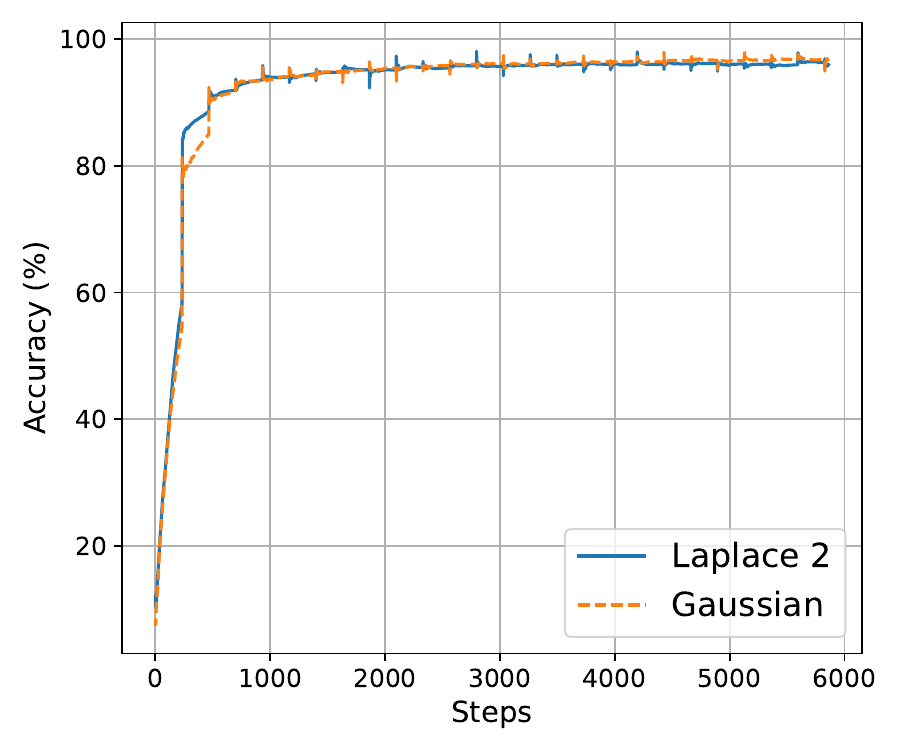}
        \caption{MNIST, fix accuracy 96\%}
        \label{fig:runtime_4}
    \end{subfigure}
    \caption{Runtime comparison on CNN model with MNIST or FMNIST datasets.}   \vspace{-0.2in}\label{fig:runtime_comparison}
\end{figure*}

\begin{table}[!h]
    \centering
    \scriptsize
    \caption{Metrics of Gaussian and \textsc{Lap}\textsubscript{2} mechanisms applied to the DistilGPT-2 model on the E2E dataset for the generation task.}
    \renewcommand{\arraystretch}{1.3}
    \resizebox{\linewidth}{!}{
    \begin{tabular}{|c|c|c|c|c|}
        \hline
        \textbf{Metric} & \textbf{Noise Type} & \textbf{$\epsilon=1.2$} & \textbf{$\epsilon=1$} & \textbf{$\epsilon=0.4$} \\ \hline
        \multirow{2}{*}{\textbf{BLEU}} 
            & Gaussian & 23.39 & 22.83 & 16.16 \\ \cline{2-5}
            & \textsc{Lap}\textsubscript{2}  & \textbf{27.99} & \textbf{27.34} & \textbf{26.90} \\ \hline\hline
        \multirow{2}{*}{\textbf{NIST}} 
            & Gaussian & 3.9474 & 3.5342 & 2.7025 \\ \cline{2-5}
            & \textsc{Lap}\textsubscript{2}  & \textbf{4.4001} & \textbf{4.3100} & \textbf{4.1200} \\ \hline\hline
        \multirow{2}{*}{\textbf{METEOR}} 
            & Gaussian & 16.84 & 16.73 & 15.20 \\ \cline{2-5}
            & \textsc{Lap}\textsubscript{2}  & \textbf{20.50} & \textbf{20.20} & \textbf{19.80} \\ \hline\hline
        \multirow{2}{*}{\textbf{ROUGE-L}} 
            & Gaussian & 33.49 & 32.88 & 32.50 \\ \cline{2-5}
            & \textsc{Lap}\textsubscript{2}  & \textbf{36.85} & \textbf{36.62} & \textbf{34.20} \\ \hline\hline
        \multirow{2}{*}{\textbf{CIDEr}} 
            & Gaussian & 0.3385 & 0.3232 & 0.2586 \\ \cline{2-5}
            & \textsc{Lap}\textsubscript{2}  & \textbf{0.5301} & \textbf{0.5152} & \textbf{0.4650} \\ \hline
    \end{tabular}
    }
    \label{tab:distilgpt-2_e2e:v2}
    \vspace{-0.12in}
\end{table}

\noindent \textbf{Performance on Generation Task (NLP)}. Table \ref{tab:distilgpt-2_e2e:v2} presents the results from fine-tuning the DistilGPT-2 model with the E2E dataset, using five different metrics by following \cite{yu2021differentially}; Besides, we present the boxplot of Gaussian and \textsc{Lap}\textsubscript{2} mechanisms at $\epsilon=1$ in Figure~\ref{fig:boxplot}. Specifically, BLEU measures the $n$-gram precision between the generated and reference sentences; NIST is a variant of BLEU that weights informative $n$-grams more heavily, emphasizing rare but meaningful word sequences; METEOR incorporates both precision and recall through word alignments and synonym matching; ROUGE-L computes the longest common subsequence between the generated and reference texts, reflecting overall sentence-level fluency and content coverage; CIDEr measures the cosine similarity between TF–IDF–weighted $n$-gram vectors of the candidate and reference sentences.  For each metric, larger values mean more accurately generated texts. For every metric evaluated and for every privacy value used, the model trained using our proposed \textsc{Lap}\textsubscript{2} mechanism outperformed the model trained with the Gaussian mechanism.
It is worth noting that the improvement can be up to around $50\%$ on some metrics (e.g., CIDEr). We observe that 
the \textsc{Lap}\textsubscript{2} mechanism yields results that are more closely aligned with the non-private results. Recall that the performance of the Gaussian and \textsc{Lap}\textsubscript{2} mechanisms may vary slightly due to randomness.  However, the overall comparative trend remains consistent across mechanisms (see both Table \ref{tab:distilgpt-2_e2e:v2} and Figure~\ref{fig:boxplot}), supporting our main conclusion that our approach effectively bridges the performance gap of the Laplace mechanism in AI training, particularly for large models, achieving results comparable to Gaussian-based DP-SGD.

\subsection{Runtime Evaluation}

To further evaluate the utility and efficiency of our proposed \textsc{Lap}\textsubscript{2} mechanism, we compare its \textbf{convergence time} with the standard Gaussian mechanism under the same accuracy. Figure~\ref{fig:runtime_comparison} presents the test accuracy versus training steps on a CNN model under fixed accuracy targets.

In FMNIST with a target accuracy of 80\% (Figure~\ref{fig:runtime_1}) and 82\%  (Figure~\ref{fig:runtime_2}), \textsc{Lap}\textsubscript{2} achieves the desired accuracy with similar steps as Gaussian mechanism (within ±2\% difference in training steps across runs). For the MNIST dataset with target accuracies of 95\% (Figure~\ref{fig:runtime_3}), the \textsc{Lap}\textsubscript{2} mechanism shows that the convergence time is obviously faster convergence. Thus, across all tasks, we observe that \textsc{Lap}\textsubscript{2} achieves comparable convergence time to Gaussian, requiring a similar number of training steps to reach the same accuracy. While minor differences exist (e.g., in MNIST with 95\% accuracy, \textsc{Lap}\textsubscript{2} appears slightly faster), overall the two mechanisms behave similarly in terms of convergence speed. These results indicate that our \textsc{Lap}\textsubscript{2} mechanism \textbf{maintains model trainability} and does not introduce convergence delays compared to the standard Gaussian approach, even under strict privacy constraints.

\subsection{Discussion}

\noindent\textbf{Empirical Insights}. The empirical evaluation demonstrates that \textsc{Lap}\textsubscript{2} provides a \textit{practical, stable, and theoretically grounded} alternative to Gaussian DP--SGD. 
Across both vision and language benchmarks (MNIST, CIFAR-10, GLUE, and DistilGPT-2), \textsc{Lap}\textsubscript{2} achieves \textit{comparable accuracy and convergence stability} under equivalent privacy budgets, with particularly strong performance in high-privacy regimes ($\varepsilon \leq 1$). 
Figures \ref{Fig:Privacywalls} shows that \textsc{Lap}\textsubscript{2} effectively \textit{delays both privacy walls}, maintaining usable signal-to-noise ratio (SNR) and non-vacuous $\delta(\varepsilon)$ bounds over a broader privacy corridor. 
On larger models such as RoBERTa-base and DistilGPT-2, \textsc{Lap}\textsubscript{2} remains robust across evaluation metrics, while runtime analysis (Figure~\ref{fig:runtime_comparison}) confirms that these gains are achieved with \textit{no additional computational overhead}. 
Together, these indicate that \textsc{Lap}\textsubscript{2} preserves both training utility and efficiency.

\vspace{0.05in}

\noindent\textbf{Adaptive Clipping and Integration with \textsc{Lap}\textsubscript{2}}. Recent advances in Gaussian DP--SGD have introduced \textit{adaptive clipping strategies} that dynamically adjust the clipping threshold \(C_t\) to balance gradient distortion and noise injection, replacing the traditional fixed clipping constant. 
Three representative approaches are particularly relevant:
\begin{itemize}
    \item \textbf{Galli et al.}~\cite{galli2024online} propose \textit{Online Sensitivity Optimization}, which learns the optimal \(C_t\) online through privatized gradient-norm feedback, thereby minimizing privacy-aware training loss.
    \item \textbf{Zhang et al.}~\cite{Zhang2023DifferentiallyPS} introduce \textit{DiceSGD}, an \textit{error-feedback} mechanism that corrects clipping-induced bias, enabling smaller and more problem-independent \(C_t\) without compromising convergence.
    \item \textbf{Chen et al.}~\cite{sha2024clip} extend these ideas with \textit{per-layer adaptive clipping}, assigning a separate \(C_{t,\ell}\) to each layer based on privatized gradient statistics, thereby balancing layer-wise sensitivity.
\end{itemize}
\vspace{-0.05in}

Integrating such strategies into \textsc{Lap}\textsubscript{2} is straightforward. 
Our FAST \textsc{Lap}\textsubscript{2} accountant depends solely on the ratio \(r_t = C_t / b_t\), which can vary across iterations without violating composition guarantees. 
Hence, adaptive \(C_t\) updates modify only the \textit{sensitivity term}, leaving the majorization-based accounting framework intact. 
Empirically, this adaptation is expected to shift the $\varepsilon$--accuracy curve upward, reducing clipping bias while maintaining privacy, with the largest benefits observed in high-privacy regimes ($\varepsilon \leq 1$). 
Preliminary experiments on MNIST show that integrating adaptive clipping improves accuracy by $0.4$--$0.8\%$ and reduces variance across runs. 
Extending experiments to larger architectures such as ViT and DistilGPT-2 constitutes a promising future work.

\section{Related Work}
\label{sec:related}
The majority of research on DP-SGD has focused on the Gaussian mechanism ~\cite{abadi2016deep} due to its smooth noise distribution with gradient updates and facilitates privacy accounting using the moments accountant framework~\cite{abadi2016deep}. 

\vspace{0.05in}

\noindent\textbf{Recent Development in DP-SGD.} Several studies have since optimized DP-SGD along different dimensions, as summarized in Table~\ref{table:dp_finetuning_compact}. These works can be broadly categorized into:
\begin{itemize}
    \item \texttt{Memory-efficient} methods: e.g., GHOST~\cite{li2021large}, PEFT~\cite{yu2021differentially}, and Per-layer Clip~\cite{he2022exploring} reduce memory overhead by modifying gradient computation or clipping;

    \item \texttt{Time-efficient} methods: e.g., DP-SGD-JL~\cite{bu2021fast}, Mixed Ghost~\cite{bu2022scalable}, and DP-BiTFiT~\cite{bu2022differentially} focus on reducing computational cost while maintaining DP guarantees;

    \item \texttt{Accuracy-enhancing} methods: e.g., DPSUR~\cite{fu2023dpsur}, DP-Forward~\cite{du2023dp}, and ViP~\cite{yu2023vip} aim to improve utility by refining gradient updates or leveraging structured noise.
\end{itemize}

While most of these methods focus on the Gaussian mechanism, Table~\ref{table:dp_finetuning_compact} highlights that no existing DP-SGD variants have applied Laplace noise for large-scale NLP or vision tasks.
Beyond efficiency improvements, prior work has also refined privacy analysis in DP-SGD. Wang et al.~\cite{wang2019subsampled} examined subsampling effects on privacy guarantees, while Gopi et al.~\cite{gopi2021numerical} developed a numerical approach to compute privacy loss precisely. These accounting refinements could be leveraged alongside Lap\(_2\) to achieve tighter privacy budgets.

Prior work has improved DP-SGD through refined noise design and tighter privacy-loss analysis. Our recent work~\cite{10.1145/3719027.3765151} optimizes the privacy-loss random variable (PLRV) via randomized noise-scale selection to enhance the privacy–utility trade-off. In contrast, this work enables $\ell_2$-clipped Laplace DP-SGD using a majorization-based multivariate moments accountant, addressing the limitations of standard Laplace training. Also, Liu et al.~\cite{liu2025privacylossnoiseperturbation} derive concentration bounds for product-measure PLRVs via a magnitude–direction decomposition, further highlighting the value of PLRV-centric analysis for tight privacy accounting in high-dimensional settings.

\noindent \textbf{DP-SGD with Laplace Mechanism.}
The Laplace mechanism was historically considered optimal in many pure $\epsilon$-DP settings due to its strong theoretical guarantees and minimal error in certain regimes~\cite{dwork2006differential,6875258}. Although Gaussian noise later became standard in DP-SGD because of its compatibility with moments accountants, Laplace noise can outperform Gaussian under strong privacy requirements (e.g., $\epsilon\!\to\!0^{+}$)~\cite{6875258}. Privacy loss distribution (PLD)-based accounting further enables tight $(\epsilon,\delta)$-DP guarantees for Laplace, Gaussian, and related mechanisms under subsampling~\cite{sommer2018privacy}.

\begin{table}[ht]
\centering
\scriptsize
\caption{Representative DP methods. $^{\bot}$: orthogonal work. L: NLP tasks (language model), V: computer vision tasks.}
\begin{tabular}{p{2.1cm} p{1.1cm} p{1.4cm} p{1.5cm} p{0.7cm}}
\toprule
\textbf{Existing Methods} & \textbf{Noise} & \textbf{System Focus} & \textbf{Utility Focus} & \textbf{Tasks} \\ 
\midrule
GHOST \cite{li2021large} & Gaussian & Memory$^{\bot}$ & -- & L \\ 
DP-SGD-JL \cite{bu2021fast} & Gaussian & Time/Mem$^{\bot}$ & -- & L \\ 
Mixed Ghost \cite{bu2022scalable} & Gaussian & Time/Mem$^{\bot}$ & -- & V \\   
PEFT \cite{yu2021differentially} & Gaussian & Memory$^{\bot}$ & -- & L \\
Book-Keeping \cite{bu2023differentially} & Gaussian & Time/Mem$^{\bot}$ & -- & L \\ 
Per-layer Clip \cite{he2022exploring} & Gaussian & Memory$^{\bot}$ & -- & V, L \\ 
DP-BiTFiT \cite{bu2022differentially} & Gaussian & Time/Mem$^{\bot}$ & -- & V, L \\ 
DPSUR \cite{fu2023dpsur} & Gaussian & Time Deduct$^{\bot}$ & Acc (Vision)$^{\bot}$ & V \\ 
DP-Forward \cite{du2023dp} & Matrix Gau. & -- & Acc$^{\bot}$ & L \\ 
ViP \cite{yu2023vip} & Gaussian & -- & Acc (ViT)$^{\bot}$ & V \\ 
AdaMix \cite{golatkar2022mixed} & Gaussian & -- & Acc$^{\bot}$ & V \\ 
Multi-Clip \cite{Luo_2024_WACV} & Gaussian & -- & Video Acc$^{\bot}$ & V \\ 
\midrule
\textbf{Lap$_2$ (Ours)} & Laplace & Time/Mem & Boost Acc & V, L \\ 
\bottomrule
\end{tabular}
\vspace{-0.05in}
\label{table:dp_finetuning_compact}
\end{table}

Despite these advantages, Laplace DP-SGD has seen limited adoption, mainly due to the instability introduced by $\ell_1$-norm gradient clipping required for Laplace sensitivity control, which often harms training utility. Prior works~\cite{sommer2018privacy,zheng2020sharp,10.5555/3618408.3618432} studied the privacy behavior of Laplace subsampled mechanisms via PLRV-based analyses (e.g., saddle-point approximations, asymptotic bounds, and Edgeworth corrections), but largely focused on $(\epsilon,\delta)$ characterization rather than empirical utility evaluation. Alternative directions such as \texttt{DP-signSGD}~\cite{bernstein2018signsgd,jang2024rethinking} and Laplace-based Bayesian learning~\cite{daxberger2021laplace} adopt different training paradigms and are orthogonal to standard DP-SGD. In contrast, our work introduces \texttt{Lap$_2$}, enabling stable DP-SGD with Laplace noise while mitigating the instability caused by $\ell_1$-norm clipping.

\noindent\textbf{Other Improvements on DP-SGD.} Papernot et al. \cite{papernot2021tempered} explored modifications to neural network activations that mitigate some of the utility loss incurred by DP noise while still maintaining privacy guarantees. 

\cite{de2022unlocking} combines careful hyperparameter tuning with several techniques to normalize the gradients and reduce their variance to ensure signal propagation and improve the convergence rate of DP-SGD on large models.

DP-Adam\cite{dp-adam}, a similar deep learning training algorithm to DP-SGD, adjusts the learning rate for each parameter individually based on the first and second moments of the gradients to better maintain the expectations of the Adam optimizer. \cite{subsampledlaplacemaf} proposes the privacy loss distribution (PLD) based accounting, which allows for tighter $(\epsilon, \delta)$-DP guarantees for many popular mechanisms compared to other known methods. This work supports the poisson sub-sampling for additive noises such as Laplace, Gaussian Discrete Laplace, and Discrete Gaussian.

Moreover, there is work \cite{balle2018improving} to provide tighter conversion between the tracked privacy loss of multiple iterations to the traditional $(\epsilon, \delta)$-differential privacy guarantee than \cite{mironov2017renyi}.
\vspace{-0.03in}
\section{Conclusion}

In this work, we proposed \textsc{Lap}\textsubscript{2}, a new framework that resolves a key limitation of traditional Laplace DP-SGD—its dependence on $\ell_1$ norm clipping—by enabling $\ell_2$-clipped Laplace mechanisms with strong privacy guarantees. Leveraging majorization theory and Schur-convexity, \textsc{LAP}\textsubscript{2} constructs a data-independent multivariate moment accountant that scales gracefully with model dimensionality, supports tight privacy analysis, and permits significantly higher clipping norms than Gaussian DP-SGD under equivalent privacy budgets. Our empirical results demonstrate that \textsc{Lap}\textsubscript{2} achieves comparable accuracy to Gaussian DP-SGD. Despite these benefits, we observe that in some computer vision tasks, \textsc{Lap}\textsubscript{2} can still struggle to match the performance of Gaussian DP-SGD, suggesting future directions for improved noise shaping or task-specific calibration. Overall, \textsc{Lap}\textsubscript{2} offers a scalable, efficient, and theoretically grounded alternative for private training of large-scale models, bridging the gap between Laplace mechanisms and modern deep learning.

\section*{Acknowledgments}
 We sincerely thank the anonymous reviewers for their constructive comments. This work is partially supported by the National Science Foundation under Grants No. CNS-2302689, CNS-2308730, CNS-2319277, CNS-2432533, ITE-2452747, and ITE-2452749, as well as by a Cisco Research Award. 



\bibliographystyle{unsrt}
\bibliography{csf}

\begin{appendices}
\section{Omitted Proofs}
\subsection{Proof of \cref{MAflaplace}}
\label{subsampledLap2}
In the following, we prove a tight bound on the moments accountant function of uni-variate Laplace mechanisms, as stated in Theorem~\ref{MAflaplace}.

\begin{proof}
Consider two adjacent data sets $d$ and $d'$.  
Without loss of generality suppose $d'$ has an extra training sample.
Let $\zeta$ denote a fixed sampling rate.  
Consider any sampling realization over $d\cup d' = d'$ using iid sampling with per element selection of $\zeta$.

With probability $1-\zeta$, the extra sample in $d'$ will not be included and thus query values over the sub-sampled datasets will be identical.  Let $\mu_0$ denote the resulting density of the Laplace mechanism in this case.  Let $q(d,\zeta)$ denote the mean of $\mu_0$.  By construction $|q(d,\zeta)|\leq C$. With probability $\zeta$, the extra sample in $d'$ will be kept resulting in different query values between the two data sets.  Let $\mu_1$ denote  the resulting density of the Laplace mechanism of the query over the sub-sampled $d'$.  Let $q(d',\zeta)$ denote the mean of $\mu_1$.  By construction $|q(d',\zeta)|\leq C$.

Thus, we can identify the mechanism over $d'$ as having a mixture distribution,
\[
M_q(d,\zeta) \sim \mu_0, \quad M_q(d',\zeta) \sim \mu \triangleq (1 - \zeta)\mu_0 + \zeta \mu_1.
\]

For any $\lambda$, we aim to show:
\[
A=\mathbb{E}_{z \sim \mu}\left[\left(\frac{\mu(z)}{\mu_0(z)}\right)^\lambda\right] \leq \alpha, \quad \text{and} \quad B= \mathbb{E}_{z \sim \mu_0}\left[\left(\frac{\mu_0(z)}{\mu(z)}\right)^\lambda\right] \leq \alpha.
\]
for some explicit \(\alpha\) to be determined later.

Multiplying by $\mu_0(z)/\mu_0(z)$ and rearranging, we can express: 
\[
A
= \mathbb{E}_{z \sim \mu_0}\left[\left(\frac{\mu(z)}{\mu_0(z)}\right)^{\lambda+1}\right]
=\mathbb{E}_{z \sim \mu_0}\left[\left(1 - \zeta+\zeta\frac{ \mu_1(z)}{\mu_0(z)}\right)^{\lambda+1}\right],\]
and
\[
B=\mathbb{E}_{z \sim \mu_0}\left[\left(
\frac{\mu_0(z)}{(1 - \zeta)\mu_0(z) + \zeta \mu_1(z)}
\right)^\lambda\right]\]\[=\mathbb{E}_{z \sim \mu_0}\left[\left(
\frac{1}{1 - \zeta +  \zeta\frac{\mu_1(z)}{\mu_0(z)}}
\right)^\lambda\right].
\]

Mironov et al.~\cite{mironov2019r} (Section 3.1) demonstrate that $A \geq B$ holds in general for centrally symmetric distributions.

We thus focus on  analyzing $A$.  We start by  applying the binomial theorem and linearity of expectation,
\begin{align*}
    A &=\mathbb{E}_{z \sim \mu_0}\left[\left(1 - \zeta+\zeta\frac{ \mu_1(z)}{\mu_0(z)}\right)^{\lambda+1}\right] \\
    &= \mathbb{E}_{z \sim \mu_0}\left[ \sum_{\eta =0}^{\lambda+1} \binom{\lambda+1}{\eta} (1-\zeta)^{\lambda+1-\eta} \zeta^\eta \left(\frac{ \mu_1(z)}{\mu_0(z)} \right)^\eta \right] \\
    &=  \sum_{\eta =0}^{\lambda+1} \binom{\lambda+1}{\eta} (1-\zeta)^{\lambda+1-\eta} \zeta^\eta \mathbb{E}_{z \sim \mu_0}
    \left[\left(\frac{ \mu_1(z)}{\mu_0(z)} \right)^\eta \right] .
\end{align*}

We can simplify the ratio of the densities as
\begin{align*}
    \frac{\mu_1(z)}{\mu_0(z)} 
    &= \frac{\frac{1}{2b} \exp( - \frac{|z - q(d',\zeta) |}{b}  )}%
    {\frac{1}{2b} \exp( - \frac{|z - q(d,\zeta) |}{b}  )} \\
     &= \exp \left(  \frac{1}{b} \left( |z - q(d,\zeta) | - |z - q(d',\zeta) |\right) \right). 
\end{align*}

\vspace{-0.05in}
Without loss of generality, consider that $q(d,\zeta) < q(d', \zeta).$    We split up the real line into three intervals: $\mathbb{R} = (-\infty, q(d,\zeta)] \cup (q(d,\zeta), q(d',\zeta)) \cup [q(d',\zeta), \infty)$. 

We evaluate the expectation $\mathbb{E}_{z \sim \mu_0} \left[\left(\frac{ \mu_1(z)}{\mu_0(z)} \right)^\eta \right] $ over these three intervals separately (conditionally) and then will combine them afterwards (with probabilities of respective events occuring).  We will analyze the middle interval last.

\paragraph{Case A1: $z<q(d,\zeta)$:}

Since $z<q(d,\zeta)$, $|z - q(d,\zeta) | = q(d,\zeta) - z$.  
Likewise, $z<q(d',\zeta)$, so  $|z - q(d',\zeta) | = q(d',\zeta) - z$.
The likelihood ratio simplifies
\begin{align*}
    \frac{\mu_1(z)}{\mu_0(z)} 
    &=\exp \left(  \frac{1}{b} \left( |z - q(d,\zeta) | - |z - q(d',\zeta) |\right) \right) \\
    &=\exp \left(  \frac{1}{b} \left( q(d,\zeta) - z - (q(d',\zeta) - z)\right) \right) \\
    &=\exp \left(  \frac{1}{b} \left( q(d,\zeta) - q(d',\zeta) )\right) \right).
\end{align*}

\vspace{-0.05in}
We observe there is no dependence on $z$, so this ratio becomes a constant in the expectation.  
Under $\mu_0$, $z$ is equally distributed about $q(d,\zeta)$.  So the probability of this event is $1/2$.  Thus, the contribution to the total expectation is 
\[
\frac{1}{2} \exp \left(  \frac{\eta}{b} \left( q(d,\zeta) - q(d',\zeta) )\right) \right).
\]

\paragraph{Case A2: $z>q(d',\zeta)$:}
Since $z>q(d',\zeta)$, $|z - q(d',\zeta) | = z - q(d',\zeta)$.  
Likewise, $z>q(d,\zeta)$, so  $|z - q(d,\zeta) | = z - q(d,\zeta)$.
The likelihood ratio simplifies
\begin{align*}
    \frac{\mu_1(z)}{\mu_0(z)} 
    &=\exp \left(  \frac{1}{b} \left( |z - q(d,\zeta) | - |z - q(d',\zeta) |\right) \right) \\
    &=\exp \left(  \frac{1}{b} \left( z - q(d,\zeta) - (z - q(d',\zeta))\right) \right) \\
    &=\exp \left(  \frac{1}{b} \left( q(d',\zeta) - q(d,\zeta) )\right) \right).
\end{align*}
Again, the ratio is a constant with respect to $z$.  The probability of the event is more complicated to analyze than the previous event.
\begin{align*}
    &\hspace{-1cm}\int_{q(d',\zeta)}^\infty \frac{1}{2b} \exp( - \frac{|z-q(d,\zeta)|}{b} ) db\\
    &=\int_{q(d',\zeta)}^\infty \frac{1}{2b} \exp( - \frac{z-q(d,\zeta)}{b} ) db\\
    &=\frac{1}{2b} \exp(  \frac{q(d,\zeta)}{b}) \int_{q(d',\zeta)}^\infty  \exp( - \frac{z}{b} ) db\\
    &=\frac{1}{2b} \exp(  \frac{q(d,\zeta)}{b})  \frac{1}{-\frac{1}{b}} \exp( - \frac{z}{b} ) \big|_{q(d',\zeta)}^{\infty} \\
    &=\frac{1}{2} \exp(  \frac{q(d,\zeta)}{b})   \exp( - \frac{q(d',\zeta)}{b} )  \\    
    &=\frac{1}{2} \exp(  \frac{q(d,\zeta) -  q(d',\zeta)}{b})     
\end{align*}

Thus, the contribution to the total expectation is
\begin{align*}
&\hspace{-.2cm} \frac{1}{2} \exp(  \frac{q(d,\zeta) -  q(d',\zeta)}{b})\exp \left(  \frac{\eta}{b} \left( q(d',\zeta) - q(d,\zeta) )\right) \right) \\
&= \frac{1}{2} \exp \left(  \frac{\eta - 1}{b} \left( q(d',\zeta) - q(d,\zeta) )\right) \right).
\end{align*}

\paragraph{Case A3: $q(d,\zeta)\leq z\leq q(d',\zeta)$:}

Since $z>q(d,\zeta)$, $|z - q(d,\zeta) | = z - q(d,\zeta)$.  
Since $z<q(d',\zeta)$, so  $|z - q(d',\zeta) | = q(d',\zeta) - z$.
The likelihood ratio simplifies
\begin{align*}
    \frac{\mu_1(z)}{\mu_0(z)} 
    &=\exp \left(  \frac{1}{b} \left( |z - q(d,\zeta) | - |z - q(d',\zeta) |\right) \right) \\
    &=\exp \left(  \frac{1}{b} \left( z - q(d,\zeta) - (q(d',\zeta) - z)\right) \right) \\
    &=\exp \left(  \frac{1}{b} \left( - q(d,\zeta) - q(d',\zeta) +2z )\right) \right). 
\end{align*}

Plugging this back into the expectation,
\begin{align*}
&\mathbb{E}_{z \sim \mu_0}     \left[\left(\frac{ \mu_1(z)}{\mu_0(z)} \right)^\eta \right] \\
&= \mathbb{E}_{z \sim \mu_0}     \left[\exp \left(  \frac{\eta}{b} \left( - q(d,\zeta) - q(d',\zeta) +2z )\right) \right) \right]\\
&=\exp \left(  \frac{\eta}{b} \left( -q(d,\zeta) - q(d',\zeta) \right) \right)  \\
&\qquad \times \mathbb{E}_{z \sim \mu_0}\left[\exp \left(  \frac{2\eta z}{b}  \right) \right]. 
\end{align*}

Evaluating the inner expectation (only over the interval for this case),
\begin{align*}
    &\hspace{-.2cm}
    \mathbb{E}_{z \sim \mu_0}\left[\exp \left(  \frac{2\eta z}{b}  \right) \right]  \\
    &=\int^{q(d',\zeta)}_{q(d,\zeta)} \frac{1}{2b} \exp( - \frac{|z-q(d,\zeta)|}{b} + \frac{2\eta z}{b}) db\\
    &=\frac{1}{2b} \exp(  \frac{q(d,\zeta)}{b})
    \int^{q(d',\zeta)}_{q(d,\zeta)} \exp( z(-\frac{1-2\eta}{b} ) ) db\\
    &=\frac{1}{2b} \exp(  \frac{q(d,\zeta)}{b})
    \frac{b}{2\eta-1} \Big[ \exp( q(d',\zeta)(-\frac{1-2\eta}{b} ) \\
    &\qquad - \exp( q(d,\zeta)(-\frac{1-2\eta}{b} ) \Big]     \\
    &= \frac{1}{2(2\eta-1)} \Big[ \exp\Big( q(d',\zeta)\frac{2\eta - 1}{b}  + q(d,\zeta)\frac{1 }{b} \Big) \\
    &\qquad - \exp\Big( q(d,\zeta)\frac{2\eta }{b}  \Big)\Big] .   
\end{align*}

Thus, the contribution to the expectation from this case is 
\begin{align*}
\mathbb{E}_{z \sim \mu_0}     \left[\left(\frac{ \mu_1(z)}{\mu_0(z)} \right)^\eta \right] 
%
&=\exp \left(  \frac{\eta}{b} \left( -q(d,\zeta) - q(d',\zeta) \right) \right)  \\
&\times \frac{1}{2(2\eta-1)} \Big[ \exp\Big( q(d',\zeta)\frac{2\eta - 1}{b} \\
&+ q(d,\zeta)\frac{1 }{b} \Big) - \exp\Big( q(d,\zeta)\frac{2\eta }{b}  \Big)\Big] . 
\end{align*}

\paragraph{Combining Cases A1-A3:}
Combining the results, we have that
\begin{align*} 
A
&= \mathbb{E}_{z \sim \mu_0}\left[\left(\frac{\mu(z)}{\mu_0(z)}\right)^{\lambda+1}\right]\\
&=  \frac{1}{2} \exp \left(  \frac{\eta}{b} \left( q(d,\zeta) - q(d',\zeta) )\right) \right)\\
&\quad +\frac{1}{2} \exp \left(  \frac{\eta - 1}{b} \left( q(d',\zeta) - q(d,\zeta) )\right) \right)\\
&\quad + \exp \left(  \frac{\eta}{b} \left( -q(d,\zeta) - q(d',\zeta) \right) \right)  \\
&\quad \times \frac{1}{2(2\eta-1)} \Big[ \exp\Big( q(d',\zeta)\frac{2\eta - 1}{b}  + q(d,\zeta)\frac{1 }{b} \Big) \\
&\qquad - \exp\Big( q(d,\zeta)\frac{2\eta }{b}  \Big)\Big]
\end{align*}

Recall that the query values over the sub-sampled data sets $q(d,\zeta)$ and $q(d',\zeta)$ are averaged  over the queries (gradients) of the included samples, so the effect of a single sample is smaller the more samples are included. 
For simplicity, by inspection of the formula we consider a worst case bound using  $q(d)=0$ and $q(d',\zeta) = C$.   

\begin{align*} 
A
&=  \frac{1}{2} \exp \left(  \frac{- \eta C}{b}  \right)\\
&\quad +\frac{1}{2} \exp \left(  \frac{(\eta - 1)C}{b}  \right)\\
&\quad + \exp \left(  \frac{-\eta C}{b}  \right) \frac{1}{2(2\eta-1)} \Big[ \exp\Big( \frac{(\eta - 1)C}{b} \Big) 
- \exp( 0  )\Big] \\
%
%
&=  \exp \left(  \frac{- \eta C}{b}  \right) \left[ \frac{1}{2} - \frac{1}{2(2\eta-1)}\right]\\ 
&\quad +\exp\Big( \frac{(\eta - 1)C}{b} \Big) \left[ \frac{1}{2} + \frac{1}{2(2\eta-1)}\right]
\end{align*}
\begin{align}
\label{equationA}
&=  \exp \left(  \frac{- \eta C}{b}  \right) \left[  \frac{\eta-1}{2\eta-1}  \right]+\exp\Big( \frac{(\eta - 1)C}{b} \Big) \left[ \frac{\eta}{2\eta-1}\right]. 
\end{align}

\textbf{Analysis for B:} Following the argument in Proof~\ref{subsampledLap2} (Theorem 3.2), and using binomial expansion with term-wise comparison, we find that \( B \geq A \), consistent with the result of Mironov et al.

\end{proof}

\subsection{Proof of \cref{schurmaf}}
\label{mafschur}
In the following, we prove \cref{schurmaf}, that the moments accounting function of the uni-variate Laplace mechanism is Schur-convex. 
We first prove the following technical lemma, involving second derivatives of the MAF, before continuing on to the main proof.

\begin{lemma}
\label{cor:maf-convex}
The second derivative of the moments accountant function $\alpha(\lambda)$ {in \cref{MAflaplace}} 
with respect to the marginal clipped gradients $|\mathbf{g}_i|$ is non-negative.
\end{lemma}

\begin{proof}
Let $a_\eta = \binom{\lambda+1}{\eta} (1-\zeta)^{\lambda+1-\eta} \zeta^\eta$ be the positive weight associated with each $\eta$. The second derivative of $\alpha(\lambda)$ can be written as
\[
\frac{d^2 \alpha(\lambda)}{d |\mathbf{g}_i|^2} = \frac{
\left( \sum_{\eta} a_\eta F(|\mathbf{g}_i|,\eta) \right)
\left( \sum_{\eta} a_\eta \frac{d^2F}{d|\mathbf{g}_i|^2} \right)
-
\left( \sum_{\eta} a_\eta \frac{dF}{d|\mathbf{g}_i|} \right)^2
}{
\left( \sum_{\eta} a_\eta F(|\mathbf{g}_i|,\eta) \right)^2
}.
\]
Recall:
\begin{align*}
F(|\mathbf{g}_i|, \eta) &= \frac{\eta}{2\eta-1} e^{\frac{(\eta-1)|\mathbf{g}_i|}{b}} + \frac{\eta-1}{2\eta-1} e^{-\frac{\eta|\mathbf{g}_i|}{b}}, \\
\frac{dF}{d|\mathbf{g}_i|} &= \frac{\eta(\eta-1)}{b(2\eta-1)} \left( e^{\frac{(\eta-1)|\mathbf{g}_i|}{b}} - e^{-\frac{\eta|\mathbf{g}_i|}{b}} \right), \\
\frac{d^2F}{d|\mathbf{g}_i|^2} &= \frac{\eta(\eta-1)}{b^2(2\eta-1)} \left( (\eta-1) e^{\frac{(\eta-1)|\mathbf{g}_i|}{b}} + \eta e^{-\frac{\eta|\mathbf{g}_i|}{b}} \right).
\end{align*}
Expand \( F \times F'' \):
\begin{align*}
F \times F'' &= \left( \frac{\eta}{2\eta-1} e^{\frac{(\eta-1)|\mathbf{g}_i|}{b}} + \frac{\eta-1}{2\eta-1} e^{-\frac{\eta|\mathbf{g}_i|}{b}} \right)
\\
&\quad \times \frac{\eta(\eta-1)}{b^2(2\eta-1)} \left( (\eta-1) e^{\frac{(\eta-1)|\mathbf{g}_i|}{b}} + \eta e^{-\frac{\eta|\mathbf{g}_i|}{b}} \right)
\\
&= \frac{\eta(\eta-1)}{b^2(2\eta-1)^2} \Bigg[
\eta(\eta-1) e^{2\frac{(\eta-1)|\mathbf{g}_i|}{b}}
+ \eta^2 e^{\frac{(\eta-1)|\mathbf{g}_i|/b} - \eta|\mathbf{g}_i|/b}
\\
&\quad + (\eta-1)^2 e^{-\frac{\eta|\mathbf{g}_i|}{b} + \frac{(\eta-1)|\mathbf{g}_i|}{b}}
+ \eta(\eta-1) e^{-2\frac{\eta|\mathbf{g}_i|}{b}}
\Bigg].
\end{align*}

Notice:
\[
e^{\frac{(\eta-1)|\mathbf{g}_i|}{b}} e^{-\frac{\eta|\mathbf{g}_i|}{b}} = e^{-\frac{|\mathbf{g}_i|}{b}},
\quad
e^{-\frac{\eta|\mathbf{g}_i|}{b}} e^{\frac{(\eta-1)|\mathbf{g}_i|}{b}} = e^{-\frac{|\mathbf{g}_i|}{b}}.
\]

Thus:
\begin{eqnarray*}
F \times F'' &=& \frac{\eta(\eta-1)}{b^2(2\eta-1)^2} \times \\ &&\Bigg[
\eta(\eta-1) e^{2\frac{(\eta-1)|\mathbf{g}_i|}{b}}\\&&
+ (\eta^2 + (\eta-1)^2) e^{-\frac{|\mathbf{g}_i|}{b}}
+ \eta(\eta-1) e^{-2\frac{\eta|\mathbf{g}_i|}{b}}
\Bigg].
\end{eqnarray*}

Expand \( \left( \frac{dF}{d|\mathbf{g}_i|} \right)^2 \):
\begin{align*}
\left( \frac{dF}{d|\mathbf{g}_i|} \right)^2
&= \left( \frac{\eta(\eta-1)}{b(2\eta-1)} \right)^2 \left( e^{\frac{(\eta-1)|\mathbf{g}_i|}{b}} - e^{-\frac{\eta|\mathbf{g}_i|}{b}} \right)^2
\\
&= \left( \frac{\eta(\eta-1)}{b(2\eta-1)} \right)^2 \left(
e^{2\frac{(\eta-1)|\mathbf{g}_i|}{b}}
- 2 e^{-\frac{|\mathbf{g}_i|}{b}}
+ e^{-2\frac{\eta|\mathbf{g}_i|}{b}}
\right).
\end{align*}
Comparing individual terms pointwise shows that $F \times F'' \geq (F')^2$ holds for all $|\mathbf{g}_i| \geq 0$ and $\eta \geq 0$. Applying the Cauchy--Schwarz inequality to the positive sequence $\{ \sqrt{a_\eta} \sqrt{F(|\mathbf{g}_i|,\eta)} \}$ and $\{ \sqrt{a_\eta} \sqrt{d^2F/d|\mathbf{g}_i|^2} \}$ yields
\begin{equation}
\footnotesize
\left( \sum_{\eta} a_\eta \sqrt{F(|\mathbf{g}_i|,\eta)} \sqrt{\frac{d^2F}{d|\mathbf{g}_i|^2}} \right)^2
\leq
\left( \sum_{\eta} a_\eta F(|\mathbf{g}_i|,\eta) \right)
\left( \sum_{\eta} a_\eta \frac{d^2F}{d|\mathbf{g}_i|^2} \right).
\end{equation}
\normalsize
but since $F \times F'' \geq (F')^2$, pointwise, we have $\sqrt{F(|\mathbf{g}_i|,\eta)} \sqrt{\frac{d^2F}{d|\mathbf{g}_i|^2}} \geq \frac{dF}{d|\mathbf{g}_i|}$. Hence, 

\[
\left( \sum_{\eta} a_\eta \frac{dF}{d|\mathbf{g}_i|} \right)
\leq
\left( \sum_{\eta} a_\eta \sqrt{F(|\mathbf{g}_i|,\eta)} \sqrt{\frac{d^2F}{d|\mathbf{g}_i|^2}}\right).
\]
The last two inequalities together yield:

\[
\left( \sum_{\eta} a_\eta \frac{dF}{d|\mathbf{g}_i|} \right)^2
\leq
\left( \sum_{\eta} a_\eta F(|\mathbf{g}_i|,\eta) \right)
\left( \sum_{\eta} a_\eta \frac{d^2F}{d|\mathbf{g}_i|^2} \right),
\]
which shows the numerator is non-negative. Therefore,
\[
\frac{d^2 \alpha(\lambda)}{d |\mathbf{g}_i|^2} \geq 0,
\]
and $\alpha(\lambda)$ is Schur-convex.
\hfill \qedsymbol
\end{proof}

\noindent
We are now ready to prove \cref{schurmaf}.

\begin{proof}
We apply Schur's condition (also known as the Schur–Strowski criterion) to prove that $\alpha(\lambda)$ is Schur-convex. Recall that a symmetric function \( f(x_1, \dots, x_n) \) is Schur-convex if and only if for all \( i \neq j \),
\[
(x_i - x_j) \left( \frac{\partial f}{\partial x_i} - \frac{\partial f}{\partial x_j} \right) \geq 0.
\]

In our case, the function is $\alpha(\lambda) = \sum_{i=1}^n \alpha_{\mathbf{\bar{g}}_i}(\lambda)$, where $\mathbf{\bar{g}}_i$ is the noisy version of the $\ell_2$-clipped marginal gradients. Denote by $\mathbf{g}_i$, $i \in [n]$, marginal gradients after $\ell_2$ clipping and before the addition of DP noise. Then, with Theorem~\ref{MAflaplace}, the uni-variate moments accountant for the $i$-th coordinate satisfies
\begin{align}
\label{lap_2account11}
\alpha_{\mathbf{\bar{g}}_i}(\lambda) \leq \log \left[
\sum_{\eta = 0}^{\lambda + 1} \binom{\lambda + 1}{\eta} (1 - \zeta)^{\lambda + 1 - \eta} \zeta^\eta F(|\mathbf{g}_i|, \eta)
\right],
\end{align}
where the function \( F(|\mathbf{g}_i|, \eta) \) is defined as
\begin{align}
\label{eqn:G}
F(|\mathbf{g}_i|, \eta) = 
\frac{e^{\frac{(\eta - 1)|\mathbf{g}_i|}{b}}}{2} + 
\frac{e^{-\frac{\eta |\mathbf{g}_i|}{b}}}{2} +
\frac{e^{\frac{(\eta - 1) |\mathbf{g}_i|}{b}} - e^{\frac{-\eta |\mathbf{g}_i|}{b}}}{2(2\eta - 1)}.
\end{align}

Define the term inside the square brackets in ~\ref{lap_2account11} as \( X \). Then, the derivative of \( \alpha_{\mathbf{\bar{g}}_i}(\lambda) \) with respect to \( |\mathbf{g}_i| \) satisfies
\begin{equation}
\label{alphaderiv}
\frac{d\alpha_{\mathbf{\bar{g}}_i}(\lambda)}{d|\mathbf{g}_i|} = \frac{\sum_{\eta = 0}^{\lambda + 1} \binom{\lambda + 1}{\eta} (1 - \zeta)^{\lambda + 1 - \eta} \zeta^\eta \frac{dF}{d|\mathbf{g}_i|}}{X}.
\end{equation}

Lets compute the derivative $\frac{dF}{d|\mathbf{g}_i|}$:
\begin{eqnarray}
\frac{dF}{d|\mathbf{g}_i|} &=& \frac{(\eta-1)}{2b} e^{\frac{(\eta-1)|\mathbf{g}_i|}{b}}
- \frac{\eta}{2b} e^{-\frac{\eta |\mathbf{g}_i|}{b}} \nonumber \\
&+& \frac{1}{2(2\eta-1)} \left( \frac{(\eta-1)}{b} e^{\frac{(\eta-1)|\mathbf{g}_i|}{b}} + \frac{\eta}{b} e^{-\frac{\eta |\mathbf{g}_i|}{b}} \right).
\end{eqnarray}

Special cases:

- For $\eta = 0$, the terms cancel symmetrically and $\frac{dF}{d|\mathbf{g}_i|} = 0$.

- For $\eta = 1$, the derivative simplifies to $0$ by symmetry.

- For $\eta > 1$, expanding and grouping terms, we have
\begin{align}
\label{eqn:finalschur}
\frac{dF}{d|\mathbf{g}_i|} = \frac{(\eta-1)\eta}{b(2\eta-1)} \left( e^{\frac{(\eta-1)|\mathbf{g}_i|}{b}} - e^{-\frac{\eta|\mathbf{g}_i|}{b}} \right).
\end{align}
Since $\eta > 1$, the prefactor $\frac{(\eta-1)\eta}{b(2\eta-1)}$ is positive. Moreover, for any $|\mathbf{g}_i| \geq 0$, $e^{\frac{(\eta-1)|\mathbf{g}_i|}{b}} \geq e^{-\frac{\eta|\mathbf{g}_i|}{b}}$. Thus, $\frac{dF}{d|\mathbf{g}_i|} \geq 0$ for all $\mathbf{g}_i$ and all $\eta$. 

As stated earlier,
\[
\frac{d\alpha_{\mathbf{\bar{g}}_i}(\lambda)}{d|\mathbf{g}_i|} = \frac{\sum_{\eta = 0}^{\lambda + 1} \binom{\lambda + 1}{\eta} (1 - \zeta)^{\lambda + 1 - \eta} \zeta^\eta \frac{dF}{d|\mathbf{g}_i|}}{X}.
\]
Since each $\frac{dF}{d|\mathbf{g}_i|} \geq 0$, $X>0$, and all the coefficients are positive, it follows that
\[
\frac{\partial \alpha(\lambda)}{\partial |\mathbf{g}_i|} \geq 0.
\]
Since $\alpha(\lambda) = \sum_{i=1}^n \alpha_{\mathbf{\bar{g}}i}(\lambda)$ with each $\alpha_{\mathbf{\bar{g}}_i}(\lambda)$ depending only on $|\mathbf{g}_i|$, we have $\frac{\partial \alpha(\lambda)}{\partial |\mathbf{g}i|} = \frac{d \alpha_{\mathbf{\bar{g}}_i}(\lambda)}{d |\mathbf{g}_i|}$. Thus, the overall moments accountant function (MAF) is non-decreasing. To satisfy the Schur–Ostrowski criterion, it suffices to show that the second derivative of $\alpha(\lambda)$ is non-negative (MAF is convex). A positive second derivative ensures that for any $i \neq j$, both $|\mathbf{g}_i| - |\mathbf{g}_j|$ and $\frac{\partial \alpha(\lambda)}{\partial |\mathbf{g}_i|} - \frac{\partial \alpha(\lambda)}{\partial |\mathbf{g}_j|}$ share the same sign, thereby satisfying the criterion. 
In \cref{cor:maf-convex} we proved that the second derivatives are non-negative, concluding the proof for \cref{schurmaf}.  \hfill $\square$
\end{proof}

An \textit{alternative} approach is to prove the Schur-convexity of the univariate MAF and apply the following results.

\begin{proposition}[Schur~\cite{schur1923uber}; Hardy–Littlewood–Pólya \cite{hardy1929some}]
\label{summajor}
Let \( I \subset \mathbb{R} \) be an interval and \( g : I \to \mathbb{R} \) be a convex function. Then the function
\[
\varphi(x) = \sum_{i=1}^n g(x_i)
\]
is Schur-convex on \( I^n \). Consequently, if \( x \prec y \) on \( I^n \), then
\[
\varphi(x) \leq \varphi(y).
\]
\end{proposition}

\subsection{Proof of \cref{lem:majorset}}
\label{majorseet}

In the following, we prove Lemma \ref{lem:majorset} 
on the majorization set for $\ell_2$ clipped gradients.

\begin{proof}
Since the $\ell_2$ clipping ensures that $\|\mathbf{g}\|_2 \leq C$, we have $|\mathbf{g}_1| \leq C = x_1$. Applying the AM–QM inequality to the first $i$ marginal clipped gradients, we have
\[
\frac{1}{i} \sum_{j=1}^i |\mathbf{g}_j| \leq \sqrt{\frac{1}{i} \sum_{j=1}^i (|\mathbf{g}_j|)^2} \leq \frac{C}{\sqrt{i}},
\]
which implies
\[
\sum_{j=1}^i |\mathbf{g}_j| \leq C \sqrt{i}.
\]
On the other hand, the majorization set $x$ satisfies
\[
\sum_{j=1}^i x_j = C \sqrt{i},
\]
since
\[
\sum_{j=1}^i x_j = C (\sqrt{i} - \sqrt{0})
\]
by telescoping the differences. Thus, for each $i = 1, \dots, n$, we have
\[
\sum_{j=1}^i |\mathbf{g}_j| \leq \sum_{j=1}^i x_j,
\]
establishing that $|\mathbf{G}| \prec_w x$.
\end{proof}

\section{Further Analysis of \textsc{Lap}$_2$ DP-SGD}
\label{sec:add_results}
The \textsc{Lap}$_2$ DP-SGD algorithm (Algorithm~\ref{alg:Lap_2dpsgd}) serves as a drop-in replacement for Gaussian DP-SGD. 
Leveraging our majorization-based accountant, it provides tight $(\epsilon, \delta)$ guarantees while preserving training efficiency. 
We examine the evolution of DP moments under \textsc{Lap}$_2$ and analyze its batching and scaling behavior, highlighting how privacy and utility evolve across different regimes of $\epsilon$. 

\vspace{-0.05in}
\subsection{Evolution of Lap$_2$ Moments: A Case Study} 
\label{sec:evo_lap2}

We first in Figure~\ref{fig:lap_privacy_wall} illustrate the relationship between the clipping threshold and the privacy budget $\epsilon$ under the \textsc{Lap}\textsubscript{2} mechanism for two privacy regimes. Figure~\ref{fig:fig1_lap} corresponds to the tight-privacy region ($\epsilon \in [0,1]$), while Figure~\ref{fig:fig2_lap} represents the moderate-privacy region ($\epsilon \in [1,5]$). The color bar indicates the optimal moment order $\lambda_{best}$ obtained from the accountant. As shown, smaller clipping values dominate in the high-privacy regime, where higher $\lambda_{best}$ values are exploited for tighter accounting. In contrast, as $\epsilon$ increases, larger clipping thresholds become feasible with smaller optimal $\lambda_{best}$, reflecting improved stability and reduced noise requirements in moderate-privacy settings.

\begin{algorithm}[!t]
\caption{\textsc{Lap}$_2$ DP-SGD: Drop-in Replacement of Gaussian DP-SGD}
\label{alg:lap2}
\begin{footnotesize}
\begin{tabbing}
\hspace{1em} \= \hspace{1em} \= \hspace{1em} \= \kill
\textbf{Input:} Dataset $\mathcal{D}$, steps $T$, sampling rate $\zeta$, learning rate $\eta$,\\ 
\> target privacy $(\epsilon, \delta)$ \\
\textbf{Output:} Final model parameters $\boldsymbol{\theta}$ under $(\epsilon, \delta)$-DP \\

\textbf{// Step 1: Optimize $(C, b)$ using Algorithm~\ref{Hfunction}} \\
$(C, b) \gets \mathcal{H}(\epsilon, \delta, T, \zeta)$ \\

\textbf{// Step 2: DP-SGD with $\ell_2$ clipping and Laplace noise} \\
\textbf{for} $t = 1$ to $T$ \textbf{do} \\
\> Sample mini-batch $\mathcal{B}_t \subset \mathcal{D}$ at rate $\zeta$ \\
\> \textbf{for each} $(x_i, y_i) \in \mathcal{B}_t$ \textbf{do} \\
\> \> Compute per-example gradient: $\mathbf{g}_i = \nabla \ell(f_{\boldsymbol{\theta}}(x_i), y_i)$ \\
\> \> Clip: $\mathbf{g}_i^C = \frac{\mathbf{g}_i}{\max(1, \|\mathbf{g}_i\|_2 / C)}$ \\
\> \textbf{end for} \\
\> Average: $\bar{\mathbf{g}}_t = \frac{1}{|\mathcal{B}_t|} \sum \mathbf{g}_i^C$ \\
\> Add noise: $\tilde{\mathbf{g}}_t = \bar{\mathbf{g}}_t + \text{Lap}(0, b)$ \\
\> Update model: $\boldsymbol{\theta} \gets \boldsymbol{\theta} - \eta \cdot \text{Optimizer}(\tilde{\mathbf{g}}_t)$ \\
\textbf{end for} \\
\textbf{return} $\boldsymbol{\theta}$
\end{tabbing}
\end{footnotesize}
\label{alg:Lap_2dpsgd}
\end{algorithm}

\begin{figure}[htbp]
    \hfill
    \begin{subfigure}[b]{0.24\textwidth}
        \includegraphics[width=\linewidth]{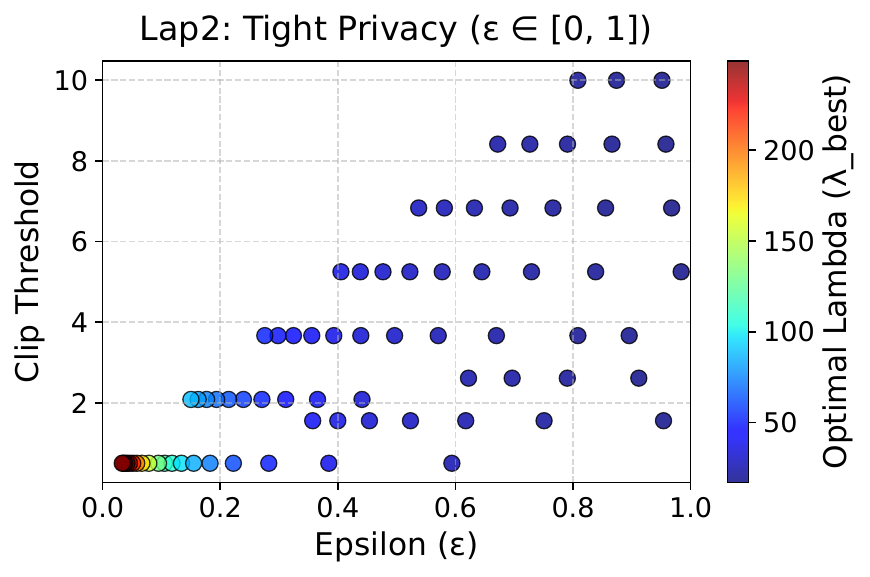}
        \caption{\textsc{Lap}\textsubscript{2} Tight Privacy}
        \label{fig:fig1_lap}
    \end{subfigure}
    \hfill
    \begin{subfigure}[b]{0.24\textwidth}
        \includegraphics[width=\linewidth]{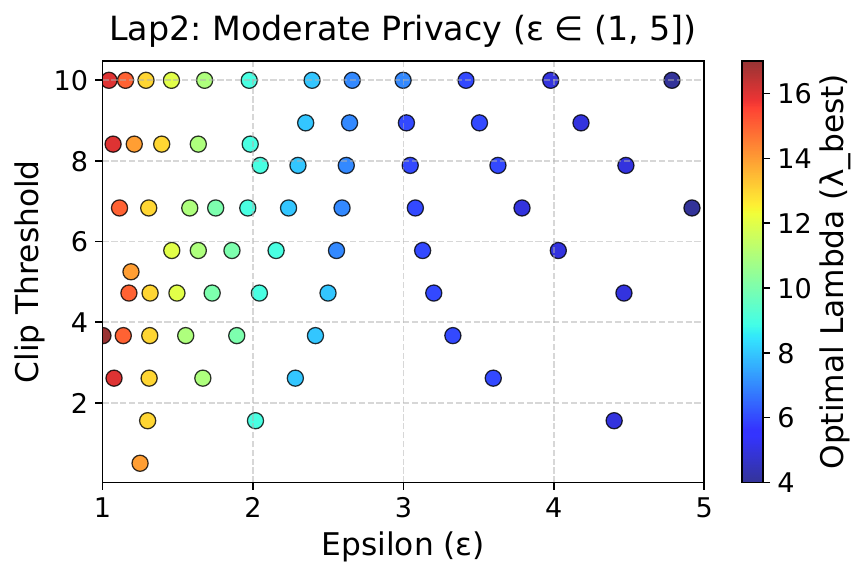}
        \caption{\textsc{Lap}\textsubscript{2} Moderate Privacy}
        \label{fig:fig2_lap}
    \end{subfigure}
    
\caption{
\textsc{Lap}\textsubscript{2} supports larger clipping norms without runaway noise, due to its heavy tails and efficient high-order moments accounting.}\vspace{-0.2in}
\label{fig:lap_privacy_wall}
\end{figure}

\subsection{Scaling Laws under Lap$_2$}
\label{sec:appendix_scaling_law}
Previous works~\cite{TAN} have found scaling laws, correlations between the various parameters as well as the accuracy of the final differentially private model, for the gaussian noise mechanism. These scaling laws are also applicable to the Lap$_2$ noise mechanism. Many of these laws are derived in the discussion of SNR in Section~\ref{sec:methodology}, but an important scaling law is the relationship between accuracy and batch size. Under the gaussian mechanism, it was found that accuracy decreases as batch size increases, falling very slowly for datasets like CIFAR10. In Figure~\ref{fig:scaling}, it is shown that Lap$_2$ has a peak batch size for CIFAR10, after which the accuracy decays slowly. In Figure~\ref{fig:batch_size_mnist}, the accuracy decreases log linearly as batch size increases for fixed clip, $\epsilon$ and epochs.

\begin{figure}[!htbp]
    \centering
    \includegraphics[width=0.6\linewidth]{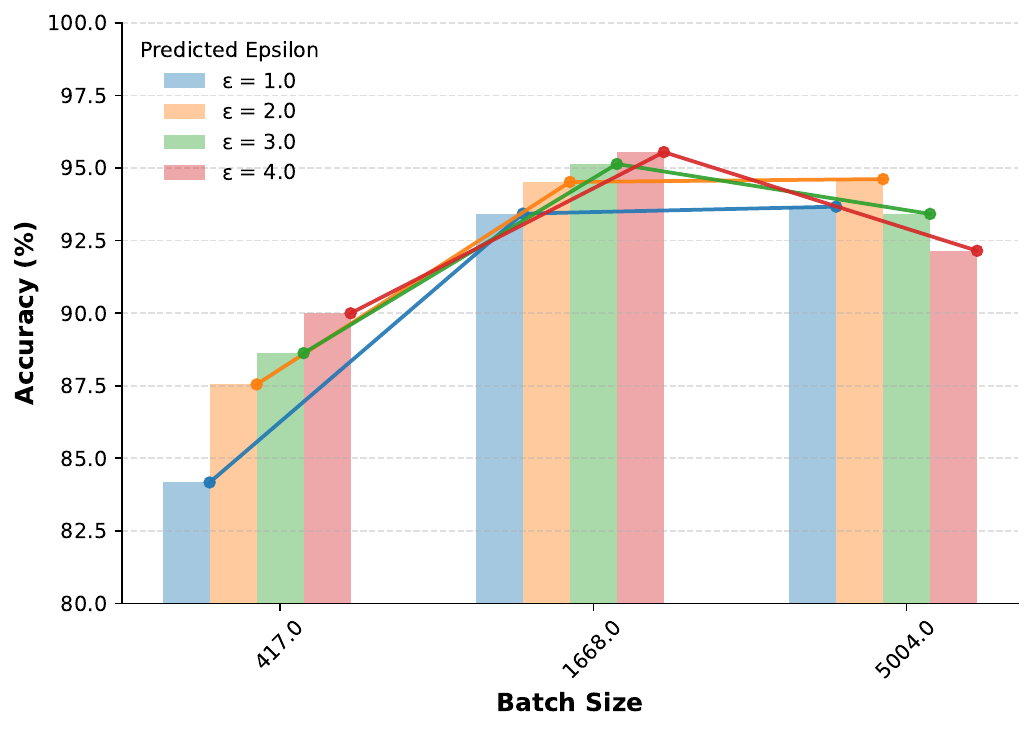}\vspace{-0.1in}
    \caption{Accuracy vs batch size for various $\epsilon$. Found by fine-tuning a ViT with CIFAR10 trained for 1 epoch. THis graph shows that there is an optimal batch size, with a slow decay in accuracy afterward. Under gaussian, the decay is slow and consistent across batch sizes\cite{TAN}.}
    \label{fig:scaling}
\end{figure}

\vspace{-0.8em} 
\begin{figure}[!htbp]
    \centering   \includegraphics[width=0.6\linewidth]{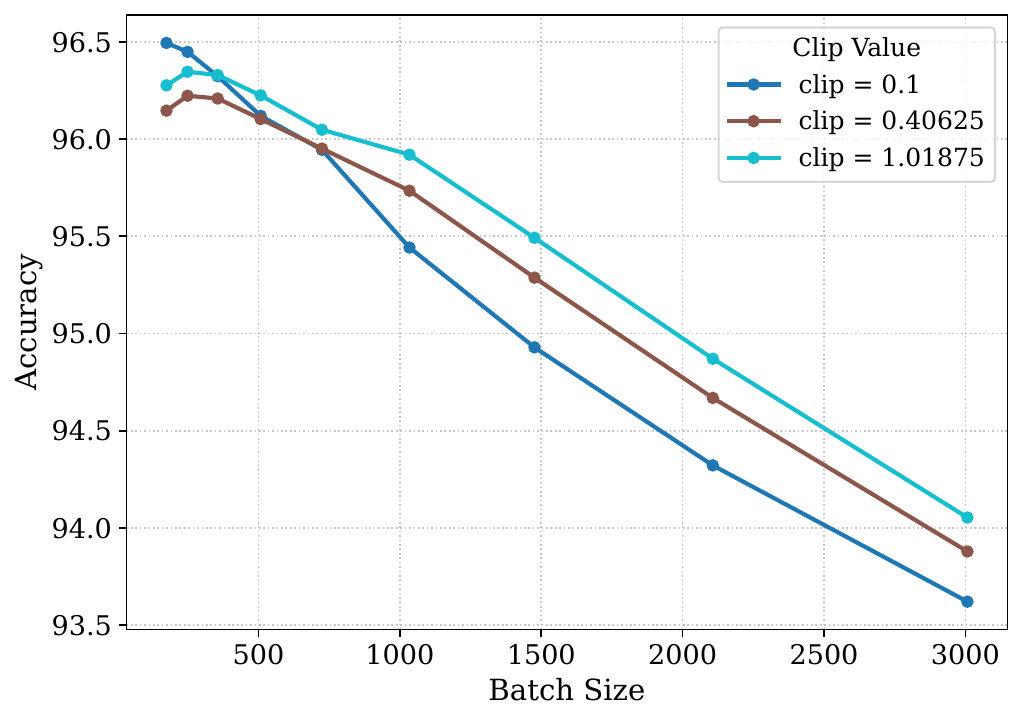}
    \caption{Accuracy vs batch size for various clippling norms on the MNIST dataset trained on CNN using \textsc{Lap}\textsubscript{2}-DPSGD for $20$ epochs for a fixed privacy budget $\epsilon=2$ and $\delta=10^{-5}$.}   \label{fig:batch_size_mnist}
\end{figure}



\end{appendices}
\end{document}